%% file: arxiv_v1.tex
\documentclass[a4paper,10pt,english]{article}
\usepackage[a4paper, margin=2cm]{geometry}
\usepackage[utf8]{inputenc}

\usepackage{mathtools,amsmath,amssymb,amsfonts}
\usepackage{mathalfa}
\usepackage{tikz}
\usetikzlibrary{calc}
\usepackage{color}
\usepackage{url}
\usepackage{soul}
\usepackage[mode=errorstop]{pstool}

\usepackage{algorithm}
\usepackage{algorithmic}

\usetikzlibrary{decorations.markings}

\usepackage{enumitem}

\renewcommand{\P}{\mathcal{P}}
\newcommand{\A}{\mathcal{A}}
\newcommand{\E}{\mathcal{E}}
\newcommand{\Eb}{\mathbb{E}}
\newcommand{\X}{\mathcal{X}}
\newcommand{\V}{\mathcal{V}}

\graphicspath{{./figures/}}

\usepackage[english]{babel}

\usepackage{tabularx}

\usepackage{xcolor}

\usepackage{verbatim}
\usepackage{xspace}

\usepackage{amsfonts}
\usepackage{mathtools}
\usepackage{amsthm}
\usepackage{amsmath}

\usepackage{amssymb}

\usepackage{setspace}

\usepackage{ifthen}
\usepackage{tikz}
\usetikzlibrary{automata,positioning, arrows}

\newtheorem{thm}{Theorem}
\newtheorem{lemma}{Lemma}
\newtheorem{prop}{Proposition}
\newtheorem{cor}{Corollary}
\theoremstyle{definition}

\newtheorem{defn}{Definition}

\newtheorem{myexample}{Example}
\newtheorem{rem}{Remark}

\newenvironment{example}[1][0]
{ 
  \ifthenelse{\equal{#1}{0}}{
  \myexample
}
{ 
  \renewcommand\themyexample{#1}\myexample
  \addtocounter{myexample}{-1}
}
}
{\endmyexample}

\newcommand{\virgolette}[1]{``#1''}

\newcommand{\R}{\mathbb{R}}{}
{}
\newcommand{\N}{\mathbb{N}}

\newcommand{\B}{\mathbb{B}}

\newcommand{\cA}{\mathcal{A}}

\newcommand{\cK}{\mathcal{K}}
\newcommand{\cKL}{\mathcal{KL}}

\newcommand{\M}{\langle M \rangle}
\newcommand{\cX}{\mathcal{X}}
\newcommand{\cF}{\mathcal{F}}

\newcommand{\cS}{\mathcal{S}}
\newcommand{\cT}{\mathcal{T}}

\newcommand{\cV}{\mathcal{V}}

\newcommand{\cZ}{\mathcal{Z}}

\newcommand{\wt}{\widetilde}

\newcommand{\dw}{\textnormal{dw}}
\DeclareMathOperator{\co}{co}
\DeclareMathOperator{\Inn}{Int}
\DeclarePairedDelimiterX{\inp}[2]{\langle}{\rangle}{#1, #2}

\title{Stability of Switched Affine Systems: Arbitrary and Dwell-Time Switching\thanks{R.J. is a FNRS honorary Research Associate. This project has received funding from the European Research Council (ERC) under the European Union's Horizon 2020 research and innovation programme under grant agreement No 864017 - L2C. R.J. is also supported by the Innoviris Foundation and the FNRS (Chist-Era Druid-net). The first and second authors contributed equally. \\The Authors are with the ICTEAM,        UCLouvain, 4 Av. G. Lema\^{i}tre, 1348 Louvain-la-Neuve, Belgium. ).
		        {\tt\small \{matteo.dellarossa,lucas.egidio,} {\tt\small raphael.jungers\}@uclouvain.be}}
}
	
\author{Matteo Della Rossa
	\and Lucas N. Egidio
	\and Raph\"ael M. Jungers 

}

\begin{document}

\maketitle
\begin{abstract}
    The dynamical behavior of switched affine systems is known to be more intricate than that of the well-studied switched linear systems, essentially due to the existence of distinct equilibrium points for each subsystem. First, under arbitrary switching rules, the stability analysis must be generally carried out with respect to a compact set with non-empty interior rather than to a singleton. We provide a novel proof technique for existence and outer approximation of attractive invariant sets of a switched affine system, under the hypothesis of global uniform stability of its linearization. On the other hand, considering dwell-time switching signals, forward invariant sets need not exist for this class of switched systems, even for stable ones. Hence, more general notions of stability/boundedness are introduced and studied, highlighting the relations of these concepts to the uniform stability of the linear part of the system under the same class of dwell-time switching signals. These results reveal the main differences and specificities of switched affine systems with respect to linear ones, providing a first step for the analysis of switched systems composed by subsystems \emph{not} sharing the same equilibrium. Numerical methods based on linear matrix inequalities and sum-of-squares programming are presented and illustrate the developed theory.  
\end{abstract}

\section{Introduction}
Switched dynamical systems provide a mathematical model for a large class of phenomena and  have been the subject of intense analysis in the past decades; for an overview, see the monographs~\cite{liberzon,ShorWirMas07}.
In this framework, given $M$ vector fields $f_1,\dots, f_M:\R^n\to \R^n$, we consider the differential equation
\begin{equation}\label{eq:SwitchIntro}
\dot x(t)=f_{\sigma(t)}(x(t)),
\end{equation}
where $\sigma:\R_+\to \{1,\dots,M\}$ is a so-called \emph{switching signal}, which selects, at each instant of time, one subsystem among $f_1,\dots, f_M$ to be active, i.e., to define how the state $x(t)$ evolves. One of the first but yet challenging problems studied in the literature is the \emph{stability} analysis of~\eqref{eq:SwitchIntro}, specifically for some particular classes of subsystems $f_1,\dots, f_M$ and/or some particular classes of switching signals $\sigma:\R_+\to \{1,\dots ,M\}$.

Regardless the specific framework (linear or non-linear subsystems, arbitrary or restricted switching signals, discrete or continuous-time evolution), a classical assumption in switched systems literature consists in considering subsystems that \emph{share the same equilibrium} at the origin, i.e., $f_1(0)=\dots=f_M(0)=0$; the stability analysis is then carried out with respect to this common equilibrium {(this is for instance the case when the subsystems $f_i$ are linear functions)}. This assumption, which is reasonable in several practical contexts, simplifies the analysis, and leads to elegant and handy Lyapunov characterizations of (several kinds of) stability notions, for an overview we refer to~\cite{liberzon,ShorWirMas07,Jung09, Wir05, geromel2006stability} and references therein. On the other hand it is often fruitful or even necessary to consider subsystems which \emph{do not} share the same equilibrium, leading to a greater modeling capacity but paying the price of a more convoluted stability analysis. A natural and interesting case is provided by switched \emph{affine} systems, thus considering maps $f_1,\dots, f_M:\R^n\to \R^n$ in~\eqref{eq:SwitchIntro} which are affine, i.e., $f_i(x)=A_ix+b_i$, for all $i\in \{1,\dots, M\}$, for some given matrices and vectors $A_i\in \R^{n\times n}$ and $b_i\in \R^n$. 

Switched affine models are often adopted in dynamical system analysis and control, mostly in the form of piecewise affine systems, to characterize or approximate complex non-linear behaviors, see \cite{kerrigan2002optimal,garulli2012survey}. { In this case, the switching events are state-dependent and subsystems are selected according to which operating region the state trajectory lies in. Beyond that, the use of affine models is also recurrent in power electronics~\cite{deaecto2010switched,egidio2020switched}, as several power converters can be conveniently represented by them.  In this setting, the switching signal is generally a control variable that must be designed to guide the state vector towards a user-defined goal.}

In our work, instead, we focus on the continuous-time switched affine systems with switching signals that are considered to be unknown exogenous disturbances/inputs and the analysis is then carried out in two distinct cases:
\begin{itemize}[leftmargin=*]
    \item the \emph{arbitrary switching case}, in which no further assumption is made on the switching signals;
    \item the \emph{dwell-time switching case}, where it is assumed the existence of a \emph{dwell-time} $\tau>0$, which represents the minimal time interval between two consecutive switching instants.
\end{itemize}

In both of these cases, in analyzing the stability of a given switched affine system, our approach requires to consider, as a first step, the corresponding \emph{linearized} switched system. More precisely, supposing that subsystems $f_i:\R^n\to \R^n$ in ~\eqref{eq:SwitchIntro} are affine and defined by matrices and vectors $(A_i,b_i)\in \R^{n\times n}\times \R^n$ with $i\in \{1,\dots,M\}$, we consider the corresponding switched linear system
\[
\dot x(t)=A_{\sigma(t)}x(t)
\]
as a surrogate in some of our analyses.
This is motivated by the fact that stability (with respect to the origin) of switched linear systems is a fundamental and mature topic in switched systems literature, see for example~\cite{liberzon,Mor96,LinAnts09} and references therein.  
Thus, considering the linearization of system~\eqref{eq:SwitchIntro}, we rely on  well-known converse Lyapunov results for switched \emph{linear} systems, see~\cite{DayMat99,MolPya89} for the arbitrary switching case, and~\cite{Wir05, ChiGugSig19} for the dwell-time case. Assuming that the corresponding linearized system is stable, the aforementioned converse Lyapunov theorems provide us with the existence of (multiple) Lyapunov functions for the linearized system. {Therefore, we generalize these results to cope with stability/boundedness properties for the original switched affine system.} As said, since the affine subsystems do not, in general, share the same equilibrium, we need to consider tailored concepts of stability. 

In the arbitrary switching  case, under the hypothesis that the corresponding linearized system is stable, we prove the existence of a compact forward invariant attractor and that all the trajectories exponentially converge to it. Moreover, some numerical techniques based on semidefinite and sum-of-squares optimization are proposed in order to provide outer approximations of the minimal forward invariant compact attractor. {The proposed analysis is inspired by the results presented in~\cite{AthJun16} for the discrete-time setting, and the continuous-time case partially introduced in~\cite{NilBosSig13}.}

When considering switching signals with minimum dwell-time, we prove that the solutions of switched affine systems are ultimately bounded, i.e., all the solutions eventually enter a compact set without leaving it afterwards, under the assumption that the linearized system is stable for the corresponding class of dwell-time signals. On the other hand, we also prove that, due to the presence of multiple equilibria, a compact forward invariant set \emph{does not} exist, in general, unlike the arbitrary switching case previously studied.
{We thus provide tailored definitions characterizing the asymptotic behavior of switched affine systems under dwell-time assumption, allowing us to underline the peculiarity of this setting with respect to the case of arbitrary switching signals.}
Numerical outer approximations of the attracting/bounding regions are provided, by means of semidefinite programming.

This manuscript is organized as follows: In Section~\ref{sec:Prelimin} we present the considered framework along with the necessary definitions and we recall the foundational Lyapunov results concerning switched linear systems. In Section~\ref{sec:Arbitrary} we study the arbitrary switching case, proving the existence of a forward invariant compact attractor, together with a numerical scheme to approximate it; the analysis for dwell-time switching signals is then developed in Section~\ref{sec:DwellTime}, {underlining the subtleties and the differences with respect to the general arbitrary switching case.} Section~\ref{sec:Concl} closes the manuscript providing some concluding remarks together with possible future directions of research.
\\
\textbf{Notation:} Given $M\in \N$, we define $\M\coloneqq\{1,\dots, M\}$. The symbols $\R_{>0}$ and $\R_{\geq 0}$ denote the sets positive and nonegative real numbers, respectively. With the notation $G:A\rightrightarrows B$ we denote a \emph{set-valued map} between $A$ and $B$. Given a set $S\subset \R^n$, the symbols $\overline S$, $\partial S$ and $\Inn(S)$ denote the \emph{closure, boundary} and \emph{interior} of $S$, respectively; $\co(S)$ denotes its \emph{convex hull}.
\section{Preliminaries}\label{sec:Prelimin}
Given $M\in \N$, consider  $\cF\coloneqq \{ (A_i, b_i)_{i\in \M}\,\vert A_i\in \R^{n \times n}, b_i\in \R^n\}$ and the continuous-time switched system
\begin{equation}\label{eq:Switchingsystems}
\dot x(t)=A_{\sigma(t)}x(t)+b_{\sigma(t)},
\end{equation}
where the \emph{switching signals} $\sigma$ are selected, in general, among the set $\cS$ defined by
\begin{equation}\label{eq:arbitrarySwitching}
\cS\coloneqq \left\{\sigma:\R_{\geq 0}\to \M\;\vert\;\;\sigma \;\text{piecewise constant} \right\}.
\end{equation}
We recall that a function $\gamma:\R_{\geq 0}\to \M$ is said to be \emph{piecewise constant} if it has a finite number of discontinuity points in any bounded subinterval of $\R_{\geq 0}$.
Moreover, without loss of generality, we suppose that signals $\sigma\in \cS$ are right-continuous.

 {Stability analysis of~\eqref{eq:Switchingsystems} under \emph{arbitrary switching signals}, can be equivalently tackled studying the differential inclusion
 \begin{equation}\label{eq:RegDiffInc}
\dot x \in F(x)\coloneqq \co\{A_i x+b_i,\,\,i \in \M\}.
\end{equation}
Indeed, it can be proven that the set of solutions to~\eqref{eq:Switchingsystems} is dense in the set of solutions of~\eqref{eq:RegDiffInc}, see for example~\cite{Filippov88},~\cite[Section 2]{ShorWirMas07} and~\cite[Theorem 2, pag.~124]{AubinCellina84} .}


It turns out that the stability/asymptotic properties of~\eqref{eq:Switchingsystems} are closely related to the properties of the corresponding \emph{linearized} switching system given by 
\begin{equation}\label{eq:DeAffinedSystem}
    \dot x(t)=A_{\sigma(t)}x(t).
\end{equation}
From now on, we call \eqref{eq:DeAffinedSystem} the \emph{linearization of system~\eqref{eq:Switchingsystems}}.
Given any $\sigma\in \cS$, let us call $\Psi_\sigma:\R_{\geq 0}\times\R^n\to \R^n$ the flow map of~\eqref{eq:Switchingsystems}, i.e., 
\[
\Psi_\sigma(t,x)\coloneqq \text{solution to~\eqref{eq:Switchingsystems}, starting at $x(0)=x$, evaluated at time $t\in \R_{\geq 0}$.}
\]
By classical linear time-varying systems literature, we have
\begin{equation}\label{eq:Solution}
\Psi_\sigma(t,x)=\Phi_\sigma(t,0)x+\int_0^t\Phi_\sigma(t,s)b_{\sigma(s)}\,ds,
\end{equation}
where $\Phi_\sigma:\R_{\geq 0}\times\R_{\geq 0}\to\R^{n\times n}$ is the state-transition matrix of the linearization~\eqref{eq:DeAffinedSystem}, see, for example,~\cite{khalil2002nonlinear}. 

We introduce the following classical definitions of stability for the linearized system~\eqref{eq:DeAffinedSystem}, which will be used in subsequent sections to characterize the behavior of~\eqref{eq:Switchingsystems}.
\begin{defn}\label{Defn:StabilityNotion}
Given any subset of switching signals $\wt \cS\subseteq \cS$ we say that the linearized system~\eqref{eq:DeAffinedSystem} is \emph{uniformly globally asymptotically stable (UGAS) on $\wt \cS$} if there exists a class $\cKL$  function\footnote{ A continuous function $\beta:\R_{\geq 0}\times \R_{\geq 0}\to \R_{\geq 0}$ is of \emph{class $\mathcal{KL}$} if $\beta(0,s)=0$ for all $s$, $\beta(\cdot,s)$ is strictly increasing for all $s$, and $\beta(r,\cdot)$ is decreasing and $\beta(r,s)\to 0$ as $s\to\infty$, for all $r$. } $\beta$ such that
\[
|\Phi_\sigma(t,0)x|\leq \beta(|x|,t)\,\,\,\forall x\in \R^n,\;\;\forall t\in \R_{\geq 0},\;\;\forall\;\sigma\in \wt \cS. 
\]
\end{defn}
We recall here a classical Lyapunov converse result for linear switching systems under arbitrary switching rules.
\begin{lemma}[Theorem 1 in~\cite{MolPya89}]\label{lemma:ConverseLyap}
The linearized switching system~\eqref{eq:DeAffinedSystem} is UGAS on $\cS$ if and only if there exist a \emph{norm} $v:\R^n\to \R_{\geq 0}$ and a scalar $\kappa>0$ such that
\begin{equation}\label{eq:LyapunovArbitraryNorm}
    v(\Phi_\sigma(t,0)x)\leq e^{-\kappa t}v(x),\,\,\forall x\in \R^n,\,\,\forall t\in \R_{\geq 0},\,\,\forall \sigma\in \cS.
\end{equation}
This in particular implies that UGAS of~\eqref{eq:DeAffinedSystem} on $\cS$ is equivalent to \emph{exponential stability}, i.e., the function $\beta\in \cKL$ in Definition~\ref{Defn:StabilityNotion} can be chosen of the form $\beta(a,t)=M e^{-\kappa t} a$ for some $M>0$.
\end{lemma}
The proof of this Lemma is provided in~\cite[Theorem 1]{MolPya89}, see also \cite[Theorem III.1]{DayMat99}.

It is usual in the switched systems setting to refine the analysis only focusing on \emph{subclasses} of $\cS$. One of the most common subclass is given by the set of dwell-time switching signals, introduced in the seminal paper~\cite{Mor96}. Formally, given a $\tau>0$, $\cS_{\dw}(\tau)$ denotes the class of \emph{dwell-time switching signals} defined by
\begin{equation}\label{eq:DwelltimeSignal}
    \cS_{\text{dw}}(\tau)\coloneqq \left\{\sigma\in \cS\;\vert\;t^\sigma_k -t^\sigma_{k-1}\geq \tau,\;\forall \;t^\sigma_k>0  \right\},
\end{equation}
where $\{t_k^\sigma\}$ denotes the set of time instants at which $\sigma$ is discontinuous, and by convention, $t^\sigma_0=0$ for all $\sigma\in \cS$. This class can be intuitively seen as the set of \virgolette{slow} switching signals, i.e., signals for which two distinct switching events cannot occur on time intervals smaller than the given threshold $\tau>0$.

As for the arbitrary switching signals case (Lemma~\ref{lemma:ConverseLyap}), the stability of~\eqref{eq:DeAffinedSystem} under any dwell time class $\cS_{\dw}(\tau)$ can also be characterized via a converse Lyapunov result, this time involving  \emph{multiple Lyapunov norms}, as recalled in the following statement.
\begin{lemma}[\cite{ChiGugSig19,Wir05}]\label{lemma:DwellTimeConverseLyapLinear}
Given $\tau>0$, the linearized system in~\eqref{eq:DeAffinedSystem} is UGAS on $\cS_{\dw}(\tau)$ if and only if there exist norms $v_1\dots v_M:\R^n\to \R_{\geq 0}$  and $\kappa>0$ such that
    \begin{align}
    v_i(e^{A_it}x)&\leq e^{-\kappa t} v_i(x),\;\;\;\;\forall x\in \R^n,\;\forall\,t\in \R_{\geq 0},\;\forall i\in \M.\label{eq:Cond1}\\
    v_i(e^{A_i\tau}x)&\leq e^{-\kappa \tau}v_j(x),\;\;\;\forall x\in \R^n,\;\forall (i,j)\in \M^2.\label{eq:Cond2}
    \end{align}
    In particular, for any $\tau>0$, UGAS of~\eqref{eq:DeAffinedSystem} on $\cS_\dw(\tau)$ is equivalent to \emph{exponential stability}.
\end{lemma}

The proof of this Lemma can be found in  \cite{ChiGugSig19} and  \cite[Corollary 6.5]{Wir05}. 

The supremum over $\kappa\in \R_{>0}$ for which the norm(s) as in Lemma~\ref{lemma:ConverseLyap} and Lemma~\ref{lemma:DwellTimeConverseLyapLinear} can be found represents the best exponential decay rate for system~\eqref{eq:DeAffinedSystem} on $\cS$ and $\cS_\dw(\tau)$, respectively; its opposite $-\kappa$ is also called the \emph{(maximal)  Lyapunov exponent} of~\eqref{eq:DeAffinedSystem} on $\cS$ and $\cS_{\dw}(\tau)$, respectively. We decided to keep this sign convention (considering the best decay rate instead of the Lyapunov exponent) for notational simplicity. The interested reader can find further discussion concerning the computation of the Lyapunov exponent for example in~\cite{ShorWirMas07}, \cite{ChiGugSig19} and references therein.

\section{Arbitrary Switching}\label{sec:Arbitrary}
In this section we first study the behavior of the switched affine system~\eqref{eq:Switchingsystems} under arbitrary switching rules, i.e., considering switching signals in the class $\cS$ given in~\eqref{eq:arbitrarySwitching}. Since, in general, the subsystems of~\eqref{eq:Switchingsystems} do not share a common equilibrium, we analyze asymptotic properties with respect to \emph{sets}. More specifically, under the hypothesis that the linearized system~\eqref{eq:DeAffinedSystem} is UGAS on $\cS$ we provide a proof of existence (and  numerical approximations) of the minimal forward invariant set. 
We first recall some definitions, characterizing properties of \emph{sets} with respect to solutions of~\eqref{eq:Switchingsystems}. 
\begin{defn}\label{defn:MainDefinitions}
Given any {subset of switching signals} $\wt \cS\subseteq \cS$ and a compact set $C\subset \R^n$, we say that:
\begin{enumerate}[leftmargin=0.8cm]
    \item $C$ is \emph{forward invariant} for~\eqref{eq:Switchingsystems} on $\wt \cS$ if, for all $x\in C$, all $\sigma\in \wt \cS$, and all $t\in \R_{\geq 0}$, it holds that
    \[
    \Psi_\sigma(t,x)\in C.
    \]
    \item A forward invariant set $C$  is \emph{minimal} if, for every forward invariant set $D$, it holds that $C\subseteq D$.
    \item $C$ is \emph{attractive} for~\eqref{eq:Switchingsystems} on $\wt \cS$ if for every $x\in \R^n$ and every $\sigma\in\wt  \cS$, we have
    \[
  |\Psi_\sigma(t,x)|_C\to 0,\,\,\,\,\text{as } \,t\to +\infty,
    \]
     where $|z|_C\coloneqq \min_{y\in C}\{|x-y|\}$ denotes the distance between a point $z\in \R^n$ and the compact set $C$,  with respect to the Euclidean norm.
    \item  $C$ is  \emph{$\cKL$-stable} for~\eqref{eq:Switchingsystems} on $\wt \cS$ if there exists a $\cKL$ function $\beta$ such that, for all $x\in \R^n$, all $\sigma\in \wt \cS$ and all $t\in \R_{\geq 0}$, it holds that
\[
|\Psi_\sigma(t,x)|_{C}\leq \beta(|x|_C,t).
\]
\end{enumerate}
\end{defn}
If a minimal forward invariant set $C$ exists,  minimality ensures its uniqueness. Moreover, it is easy to see that $\cKL$-stability implies forward invariance and attractiveness, but it is a stronger property in general (for further discussions see \cite{TeelPraly00}).
 In the following statement we prove a crucial property of forward invariant sets.
\begin{lemma}\label{lemma:convexity}
For any $t\in \R_{\geq 0}$, any $x,y\in \R^n$, any $\sigma\in \cS$, and any $\lambda\in[0,1]$ we have
\[
\Psi_\sigma(t, \lambda x+(1-\lambda)y)=\lambda \Psi_\sigma(t,x)+(1-\lambda)\Psi_\sigma(t,y).
\]
In particular given any set of signals $\wt \cS\subseteq \cS$, if $C\subset\R^n$ is forward invariant for~\eqref{eq:Switchingsystems} on $\wt S$, so is $\co\{C\}$.
\end{lemma}
\begin{proof}
{For the first part of the statement, using~\eqref{eq:Solution}, we compute:}
\[
\begin{aligned}
\Psi_\sigma(t, \lambda x+(1-\lambda)y)&=\Phi_\sigma(t,0)(\lambda x+(1-\lambda)y)+\int_0^t\Phi_\sigma(t,s)b_{\sigma(s)}\,ds\\&=
\lambda\left (  \Phi_\sigma(t,0)x+\int_0^t\Phi_\sigma(t,s)b_{\sigma(s)}\,ds\right )+&\\&\quad+(1-\lambda)\left (  \Phi_\sigma(t,0)y+\int_0^t\Phi_\sigma(t,s)b_{\sigma(s)}\,ds\right )\\&=\lambda \Psi_\sigma(t,x)+(1-\lambda)\Psi_\sigma(t,y).
\end{aligned}
\]
For the second part, consider any $\wt \cS\subseteq \cS$ and suppose that $C$ is a forward invariant set for~\eqref{eq:Switchingsystems} on $\wt S$. Consider any $x,y\in C$, any $\sigma\in \wt \cS$, any $\lambda\in [0,1]$ and any $t\in \R_{\geq 0}$, then
\[
\Psi_\sigma(t, \lambda x+(1-\lambda)y)=\lambda \Psi_\sigma(t,x)+(1-\lambda)\Psi_\sigma(t,y)\in \co\{C\},
\]
since by forward invariance of $C$, we have $\Psi_\sigma(t,x), \Psi_\sigma(t,y)\in C$.
\end{proof}
\subsection{Minimal Forward Invariant and Attractive Set}
{
In this subsection, we prove the existence of minimal forward invariant and attractive sets for~\eqref{eq:Switchingsystems} on $\cS$, i.e., under arbitrary switching sequences. Moreover, we also prove that this set is $\cKL$-stable. For that, we first need a \virgolette{theoretic outer bound}, ensuring that, under some hypotheses, forward invariant sets do exist. }
\begin{prop}\label{prop:Kvm}
Consider system~\eqref{eq:Switchingsystems}, and suppose that the linearized system  in~\eqref{eq:DeAffinedSystem} is UGAS on $\cS$. Consider a norm $v:\R^n\to \R_{\geq 0}$ and a scalar $\kappa>0$ as in Lemma \ref{lemma:ConverseLyap}. Then there exists $R>0$ such that the compact set
\begin{equation}\label{eq:InvariantLevelSetK}
    \cK_{v,R}\coloneqq \{x\in \R^n\,\,\vert\,\,v(x)\leq R\},
\end{equation}
is forward invariant for \eqref{eq:Switchingsystems} on $\cS$.
\end{prop}
\begin{proof}
Define $B_{max}\coloneqq \max_{i\in \M}\{v(b_i)\}$.
For any $x\in \R^n$, any $\sigma\in \cS$ and any $t\in \R_{\geq 0}$, compute
\[
\begin{aligned}
v(\Psi_\sigma(t,x))&=v\left( \Phi_\sigma(t,0)x+\int_0^t \Phi_\sigma(t,s)b_\sigma(s)\,ds\right )\\
&\leq v( \Phi_\sigma(t,0)x)+v\left(\int_0^t \Phi_\sigma(t,s)b_\sigma(s)\,ds\right )\\
&\leq e^{-\kappa t}v(x)+\int_0^t v (\Phi_\sigma(t,s)b_\sigma(s))\,ds\\
&\leq e^{-\kappa t}v(x)+\int_0^t e^{-\kappa(t-s)} v (b_\sigma(s))\,ds\\
&\leq e^{-\kappa t}v(x)+B_{max}\int_0^t e^{-\kappa(t-s)} \,ds\\
&=e^{-\kappa t}v(x)+B_{max}\frac{(1-e^{-\kappa t})}{\kappa},
\end{aligned}
\]
where we used~\eqref{eq:LyapunovArbitraryNorm} and the fact that, for any $\sigma\in \cS$ and any $t,s\in \R_{\geq 0}$ it holds that $\Phi_\sigma(t,s)=\Phi_{\wt \sigma}(t-s,0)$
with $\wt \sigma\in \cS$ defined by $\wt \sigma(t')=\sigma(t'+s)$.
It thus holds that if $x\in \R^n$ is such that $v(x)\geq \frac{B_{max}}{\kappa}=:R$, then $v(\Psi_\sigma(t,x))\leq v(x),$ for all $t\in \R_{\geq 0}$, concluding the proof.
\end{proof}
This statement provides a first outer bound for the minimal forward invariant set (if it exists). In the following, we show that such a set does exist. {To this aim, we introduce the set-valued map $\cK:\R_{\geq 0}\rightrightarrows \R^{n}$ defined by
\begin{equation}\label{eq:K(t)}
\begin{aligned}
\cK(t)\coloneqq \bigcup_{\sigma \in \cS}\left \{\int_0^t\Phi_\sigma(t,s)b_{\sigma(s)}\,ds\;\right \}.
\end{aligned}
\end{equation}
Equivalently, for every $t\in \R_{\geq 0}$, the set $\cK(t)$ represents the \emph{reachable set} of~\eqref{eq:Switchingsystems} at time $t$ starting at $0$, i.e.
\[
\cK(t)=\{\Psi_\sigma(t,0)\;\;\vert\;\;\sigma\in \cS\}.
\]
In what follows, studying the properties of the set-valued map $\cK:\R_{\geq 0}\rightrightarrows \R^{n}$, we prove the existence of the minimal forward invariant set.}

\begin{thm}\label{lemma:K_inftyDef}
Let us consider $\cF\coloneqq \{ (A_i, b_i)_{i\in \M}\,\vert A_i\in \R^{n \times n}, b_i\in \R^n\}$ defining a switched affine system as in \eqref{eq:Switchingsystems}. If the linearized system \eqref{eq:DeAffinedSystem} is UGAS on $\cS$, the limit
\begin{equation}
    \cK_\infty=\lim_{t\to \infty}\cK(t)\label{eq:K_inf_def}
\end{equation}
(in the Hausdorff metric on compact sets) of the sequence in~\eqref{eq:K(t)} is well-defined and equal to the minimal forward invariant set of \eqref{eq:Switchingsystems} on $\cS$.
\end{thm}
\begin{proof}
Without loss of generality, we suppose $b_1=0$; indeed, the general case is reduced to this framework by applying the translation $T:\R^{n}\to \R^{n}$ defined as $T(x)\coloneqq  x+A_1^{-1}b_1$, {where $A_1^{-1}$ exists since, by hypothesis, $A_1$ is a Hurwitz matrix. \\
\emph{Well-defineteness:} Firstly, we prove that, for all $\overline t \geq t \geq 0$, we have
$\cK(t)\subseteq \cK(\overline t)$.}
Considering any $x\in \cK(t)$, by definition, there exists a $\sigma\in \cS$, such that $x=\Psi_\sigma(t,0)$.
Now define $\wt \sigma \in \cS$ by
\[
\wt \sigma(s)\coloneqq \begin{cases}
1, \,\,\,\,\,\,\,&\text{if } s <\overline t-t,\\
\sigma(s-\overline t+t),\,\,&\text{if } s \geq \overline t-t.
\end{cases}
\]
Since $b_1=0$, it is clear that $\Psi_{\wt \sigma}(\overline t-t,0)=0$, and thus $x=\Psi_{\wt \sigma}(\overline t,0)\in \cK(\overline t)$, proving that 
$\cK(t)\subseteq \cK(\overline t)$.
By Proposition \ref{prop:Kvm}, for every $t\geq 0$, the reachable sets $\cK(t)$ are uniformly bounded { since they are included in the compact set $\cK_{v,R}$ defined in~\eqref{eq:InvariantLevelSetK}}, recalling that $0\in \cK_{v,R}$ and $\cK_{v,R}$ is forward invariant.  {Thus, the set 
\[
\cK_\infty=\overline{\bigcup_{t\geq 0}\cK(t)}
\]
 is a well-defined compact set, being the limit of an increasing sequence of compact sets, see \cite[Chapter 4.B]{RockWets98}.}\\
\emph{Forward invariance:} Consider first $x\in \bigcup_{t\geq 0}\cK(t)$ and consider any $\sigma \in \cS$, we prove that, for any $t_0\geq 0$, $\Psi_\sigma(t_0,x)\in \bigcup_{t\geq 0}\cK(t)$. Since $x\in \bigcup_{t\geq 0}\cK(t)$ there exists a $T\geq 0$ and a $\sigma_1\in \cS$ such that $x=\Psi_{\sigma_1}(T,0)$.
Defining $\wt \sigma\in \cS$ by
\[
\wt \sigma(s)\coloneqq \begin{cases}
\sigma_1(s)\,\,\,\,&\text{if }s< T,\\
\sigma(s-T)\,\,\,&\text{if } s\geq T,
\end{cases}
\]
we have $\Psi_\sigma(t_0,x)=\Psi_{\wt \sigma}(T+t_0,0)\in \cK(T+t_0)\subseteq \bigcup_{t\geq 0}\cK(t)$. Now, for the limit case where $x\in \cK_\infty=\overline{\bigcup_{t\geq 0}\cK(t)}$, for any $\sigma\in \cS$ and any $t_0\geq 0$, we want to prove that $\Psi_\sigma(t_0,x)\in \cK_\infty$. Consider a sequence $(x_k)_{k\in \N}$ such that $x_k\to x$ as $k\to \infty$ and $x_k\in \bigcup_{t\geq 0}\cK(t)$ for all $k\in \N$. We have already proven  that $\Psi_\sigma(t_0,x_k)\in \bigcup_{t\geq 0}\cK(t)$, for all $k\in \N$. Now by continuity from initial conditions (see \cite[Theorem 3.4]{khalil2002nonlinear}) we have
\[
\Psi_\sigma(t_0,x)=\lim_{k\to \infty}\Psi_\sigma(t_0,x_k),
\]
and thus $\Psi_\sigma(t_0,x)\in \overline{\bigcup_{t\geq 0}\cK(t)}=\cK_\infty$, concluding the proof.\\
\emph{Minimality:} Consider a  compact forward invariant set $C\subset \R^n$,  any initial condition $x\in C$, and the constant switching signal $\sigma(s)= 1$, for all $s\in \R_{\geq 0}$. Consider any strictly increasing sequence $(t_k)_{k\in \N}$ such that $t_k>0$ for all $k\in \N$ and $t_k\to +\infty$ as $k\to \infty$. Then $x_k=\Psi_\sigma(t_k,x)\in C$ for all $k\in \N$ by forward invariance of $C$, and since $\lim_{k\to \infty}x_k=\lim_{k\to \infty}\Psi(t_k,x)=\lim_{k\to \infty}e^{A_1t_k}x=0$, by compactness we have $0\in C$. Now by definition of forward invariance and since $\cK_\infty$ is defined as the reachable set from $0$, we conclude that $\cK_\infty\subseteq C$.
\end{proof}
An existence result similar to Theorem~\ref{lemma:K_inftyDef} was proven in~\cite{NilBosSig13}, in a slightly different setting and with a different methodology. We presented here our proof based on Lemma~\ref{lemma:ConverseLyap} since it allows us to not only prove the existence of the minimal forward invariant set $\cK_\infty$, but also to verify that it is (exponentially) $\cKL$-stable (and thus, in particular, attractive), as proven in the following proposition.
\begin{prop}\label{prop:KLstability}
Under the hypothesis of Theorem~\ref{lemma:K_inftyDef}, the compact set $\cK_\infty$ is $\cKL$-stable for~\eqref{eq:Switchingsystems} on $\cS$.
\end{prop}
\begin{proof}
Consider a scalar $\kappa>0$ and a norm $v:\R^n\to \R_{\geq 0}$ satisfying the properties in Lemma~\ref{lemma:ConverseLyap}, and define $\B_v\coloneqq \{x\in \R^n\;\vert\;v(x)\leq 1\}$, i.e., the unit ball of the norm $v$. We consider $W:\R^n\to \R_{\geq 0}$, the \emph{distance from $\cK_\infty$ with respect  to $v$}, defined by
\begin{equation}
\label{eq:Definitionv-Haudrorff}
W(x)\coloneqq \min_{y\in \cK_\infty}\{v(x-y)\}=\min\{r\in \R_{\geq 0}\;\vert\;x\in \cK_\infty+r\B_v\}.
\end{equation}
First, for any $x\in \cK_\infty$, forward invariance of $\cK_\infty$ implies that $W(\Psi_\sigma(t,x))=0$, for all $\sigma\in \cS$ and all $t\in \R_{\geq 0}$. Now consider any $x\notin \cK_\infty$, any $\sigma\in \cS$ and any $t\in \R_{\geq 0}$. By equation~\eqref{eq:Definitionv-Haudrorff}, we can decompose $x$ as $x=\hat x+y$ with $\hat x\in \cK_\infty$ and $y\in W(x)\B_v$. Computing
\[
\begin{aligned}
	W(\Psi_\sigma(t,x))&=W\left(\Phi_\sigma(t,0)(\hat x+y)+\int_0^t \Phi_\sigma(t,s)b_{\sigma(s)}\,ds\right)\\&= W\left(\Phi_\sigma(t,0)\hat x+\int_0^t \Phi_\sigma(t,s)b_{\sigma(s)}\,ds+\Phi_\sigma(t,0)y\right).
\end{aligned}
\]
Recalling again the forward invariance of $\cK_\infty$, we have $\Phi_\sigma(t,0)\hat x+\int_0^t \Phi_\sigma(t,s)b_{\sigma(s)}\,ds\in \cK_\infty$ and, by Lemma \ref{lemma:ConverseLyap}, we have $\Phi_\sigma(t,0)y\in e^{-\kappa t}W(x)\B_v$, proving that
\[
W(\Psi_\sigma(t,x))\leq e^{-\kappa t}W(x).
\]
Now by equivalence of norms in $\R^n$, it can be seen that there exist $M_1,M_2\in\R_{>0}$ such that
\[
M_1|x|_{\cK_\infty}\leq W(x)\leq M_2|x|_{\cK_\infty},\;\;\forall x\in \R^n
\]
where, we recall, $|\cdot|_{\cK_\infty}$ denotes the distance from $\cK_\infty$ \emph{with respect to the Euclidean distance}. Thus we have
\begin{equation}
|\Psi_\sigma(t,x)|_{\cK_\infty}\leq \frac{M_2}{M_1}e^{-\kappa t}|x|_{\cK_\infty},\;\;\forall x\in \R^n,\;\;\forall \sigma\in \cS,\;\;\forall t\in \R_{\geq 0},\label{eq:decay_rate_Kinf}
\end{equation}
proving Item {4.} of Definition~\ref{defn:MainDefinitions} and concluding the proof.
\end{proof}

We note that the decay rate $\kappa>0$ of the linearized system~\eqref{eq:DeAffinedSystem} with respect to the origin is somehow conserved by the switched affine system~\eqref{eq:Switchingsystems}, in the sense that $\kappa>0$ is also the exponential decay rate of~\eqref{eq:Switchingsystems} with respect to the minimal forward invariant set $\cK_\infty$, as shown in \eqref{eq:decay_rate_Kinf}. More precisely, the condition in Item~4. of Definition~\ref{defn:MainDefinitions} is satisfied in particular by a $\cKL$ function $\beta$ of the form $\beta(s,t)=Mse^{-\kappa t}$, for $M>0$, and where $\kappa>0$ is the decay rate of the linearized system~\eqref{eq:DeAffinedSystem}.
{\begin{rem}[Properties of $\cK_\infty$]
Since $\cK_\infty$ is formally defined as an infinite union of sets in \eqref{eq:K_inf_def}, it can be hard to construct it explicitly, as we will discuss in the next subsection. However, despite this non-constructive definition of $\cK_\infty$, we are able to provide some remarkable properties. First of all, given $\cF\coloneqq \{ (A_i, b_i)_{i\in \M}\,\vert A_i\in \R^{n \times n}, b_i\in \R^n\}$, let us consider the set of \emph{Filippov equilibria} defined by
\begin{equation}\label{eq:Filippov_eq}
\text{Fil}_0(\cF)\coloneqq \left \{y=-\left (\sum_{i\in\M}\overline \lambda_i A_i\right)^{\!\!-1}\!\!\!\left (\sum_{i\in\M}\overline \lambda_i b_i\right)\;\;\Big\vert\;\;\overline\lambda\in \Lambda_n \right \}=\{y\in \R^n\;\vert\;0\in F(y)\},
\end{equation}
where the set-valued map $F$ is defined in~\eqref{eq:RegDiffInc}. As discussed in Section~\ref{sec:Prelimin}, solutions to~\eqref{eq:Switchingsystems} under arbitrary switching rules are dense in the set of solutions to the differential inclusion~\eqref{eq:RegDiffInc}, and this in particular implies $\text{Fil}_0(\cF)\subseteq \cK_\infty$.
Moreover, we can ensure that the set $\cK_\infty$ is (path) connected: consider $x_1,x_2\in \cK_\infty$, and consider the constant signal defined by $\widehat \sigma(t)\equiv 1$. By forward invariance of $\cK_\infty$ it is clear that the corresponding trajectories are contained in $\cK_\infty$, more precisely $\{\Psi_{\widehat \sigma}(t,x_j)\;\vert\;t\geq 0\}\subseteq \cK_\infty$ for all $j\in \{1,2\}$. Since $\cK_\infty$ is closed and  $\lim_{t\to +\infty}\Psi_{\widehat \sigma}(t,x_1)=\lim_{t\to +\infty}\Psi_{\widehat \sigma}(t,x_2)=-A_1^{-1}b_1$, there is a path in $\cK_\infty$ connecting $x_1$ and $x_2$ (intuitively, the union of the two trajectories). We have thus proven that $\cK_\infty$ is path-connected. However, in the subsequent numerical example, we show that $\cK_\infty$ is not convex, in general. \hfill $\triangle$
\end{rem}}
\begin{myexample}\label{ex:FirstExample}
Consider a switched affine system as in~\eqref{eq:Switchingsystems} defined by
\[
A_1=\begin{bmatrix}
      	-1 & -1 \\ 
      	 0 & -1
     \end{bmatrix},~A_2=\begin{bmatrix}
      	-1 & 0 \\ 
      	 -1 & -1
     \end{bmatrix},~b_1=b_2=\begin{bmatrix}
     -1\\-1
     \end{bmatrix}
\]
The set $\text{Fil}_0(\cF)$ of the Filippov equilibria, defined in \eqref{eq:Filippov_eq}, is depicted in black in Figure~\ref{fig:counter_ex} and, as previously discussed, we know that $\text{Fil}_0(\cF)\subseteq \cK_\infty$. {Considering the convex set $\mathcal{Z}$ depicted in Figure~\ref{fig:counter_ex}, it is possible to show that, for any $x\in \partial \mathcal{Z}$ and for any $i\in \{1,2\}$ we have  $A_ix+b_i\notin \cT_\cZ(x)$, where  $\cT_\cZ(x)$ denotes the tangent cone\footnote{ Given a closed convex set $C\subset \R^n$ and $x\in \R^n$ the \emph{tangent cone to $C$ at $x$} is the set $\cT_C(x)\coloneqq \text{cl}\left(\{z\in \R^n\;\vert\;\exists\;\alpha_0>0\;\; \text{s.t. } \;x+\alpha z\in C,\;\forall \alpha\in (0,\alpha_0)\}\right)$, see \cite[Appendix A]{BroTan20} for further discussions.} to $\cZ$ at $x$. Applying the Nagumo Theorem~\cite[Chapter 4, pag. 175]{AubinCellina84}, by forward invariance of $\cK_\infty$, this implies $\text{Int}(\mathcal{Z})\cap \cK_\infty=\emptyset$. Consider the points $x_{e1},x_{e2}\in \text{Fil}_0(\cF)$, defined by $x_{ei}=-A_i^{-1}b_i$ with $i\in \{1,2\}$, i.e., the equilibria of each affine subsystem. It holds that $x^*=\frac{1}{2}x_{e1}+\frac{1}{2}x_{e2}$ is both in $\text{Int}(\mathcal{Z})$ and in $\text{co}(\text{Fil}_0(\cF))\subseteq\text{co}(\cK_\infty)$, which implies that $\cK_\infty$ is non-convex. }
\begin{figure}
    \centering
	\def\svgwidth{.6\linewidth}
    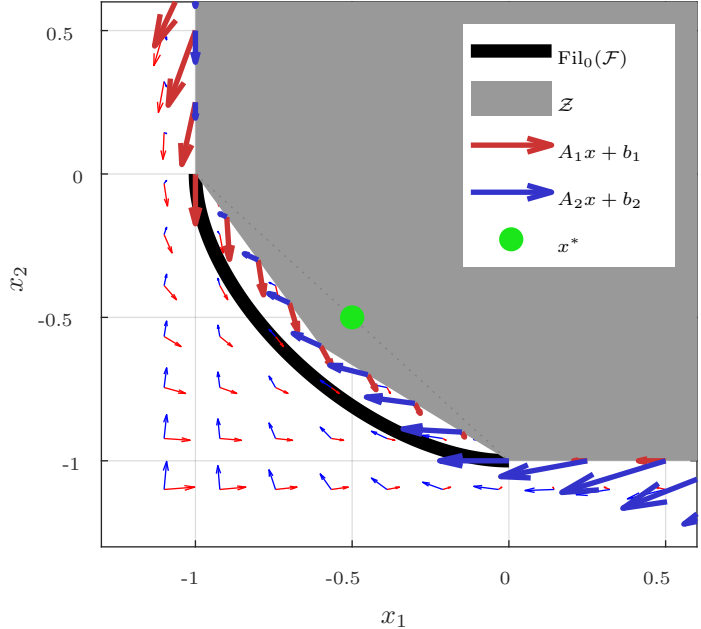

    \caption{\small Representation of the state-space for Example~\ref{ex:FirstExample} with the set of Filippov equilibria $\text{Fil}_0(\cF)$, the set $\mathcal{Z}$ for which the vector fields $A_ix+b_i,~i\in\{1,2\}$ do not point inwards, {and the point $x^*=\frac{1}{2}x_{e1}+\frac{1}{2}x_{e2}=[-0.5~-0.5]^\top\in\text{co}\{\text{Fil}_0(\cF)\}$.} This demonstrates that $\cK_\infty$ is non-convex since $x^*\in\mathcal{Z}$ implies that $x^*\notin \cK_\infty$. } 
    \label{fig:counter_ex}
\end{figure}
\end{myexample}
This example will be later recalled when we present numerical procedures to over-approximate $\cK_\infty$, {providing further insights in its geometric characterization.}

Before concluding this subsection, we provide a \virgolette{negative} result, stating that (non-trivial) forward invariant sets do not exist if the linearized system~\eqref{eq:DeAffinedSystem} is unstable. 
    \begin{prop}\label{prop:Instability}
Consider $\cF\coloneqq \{ (A_i, b_i)_{i\in \M}\,\vert A_i\in \R^{n \times n}, b_i\in \R^n\}$. If the linearized system~\eqref{eq:DeAffinedSystem} is unstable on $\cS$ (i.e., there exist $x\in \R^n$ and $\sigma \in \cS$ such that $\limsup_{t\to\infty}|\Phi_\sigma(t,0)x|=+\infty$) then there does not exist a compact forward invariant set with non-empty interior for the switched affine system~\eqref{eq:Switchingsystems} on $\cS$. Moreover, if the set of matrices $\cA=\{A_1,\dots, A_M\}\subset \R^{n\times n}$ is also \emph{irreducible}\footnote{A set of matrices $\cA\subset \R^{n\times n}$ is said to be \emph{irreducible} if $\{0\}$ and $\R^n$ are the only vector subspaces which are invariant under all matrices in $\cA$.} then there exists no compact non-singleton forward invariant set for~\eqref{eq:Switchingsystems} on $\cS$.
\end{prop}
\begin{proof}
Let us consider $x\in \R^n$ and $\sigma \in \cS$ which provide an unbounded trajectory for the linearized system~\eqref{eq:DeAffinedSystem}. Suppose by contradiction that $C\subset\R^n$ is a compact forward invariant set for~\eqref{eq:Switchingsystems}, and $\Inn(C)\neq \emptyset$. Consider any $y\in \Inn(C)$, and a $\lambda > 0$ small enough such that $z\coloneqq  y+\lambda x\in C$. Recalling the solution~\eqref{eq:Solution}, we obtain
\[
|\Psi_\sigma(t,z)-\Psi_\sigma(t,y)|=|\Phi_\sigma(t,0)z-\Phi_\sigma(t,0)y|=|\Phi_\sigma(t,0)(z-y)|=\lambda\,|\Phi_\sigma(t,0) x|,
\]
and thus 
\[
\limsup_{t\to \infty}|\Psi_\sigma(t,z)-\Psi_\sigma(t,y)|=+\infty,
\]
in contradiction to the fact that $y,z\in C$ and $C$ is a compact forward invariant set.

For the second part of the statement, recalling the results contained in~\cite[Theorem 3]{ChiMasSig12} or~\cite{ProJun15a}, since $\cA$ is irreducible and generates an unstable switched system, there exist $\rho>0$ and a norm $v:\R^n\to \R_{\geq 0}$ such that
\begin{equation}\label{eq:DestabilizeSequence}
\forall\;x\in \R^n,\;\exists\, \sigma_x\in \cS\;\;\text{such that}\;\;v(\Phi_{\sigma_x}(t,0)x)=e^{\rho t}v(x), \;\;\forall\;t\in \R_{\geq 0}.
 \end{equation}
 Suppose, by contradiction, that $C\subset\R^n$ is a compact forward invariant set and is not a singleton. Consider $y,z\in C$ with $y\neq z$ and define $x\coloneqq  z-y$. For $\sigma_x\in \cS$ defined in~\eqref{eq:DestabilizeSequence}, we have 
 \[
|\Psi_{\sigma_x}(t,x_2)-\Psi_{\sigma_x}(t,x_1)|=|\Phi_{\sigma_x}(t,0)z-\Phi_{\sigma_x}(t,0)y\;|=|\Phi_{\sigma_x}(t,0)x|.
 \] 
By equivalence of norms, this implies that $\limsup_{t\to \infty}|\Psi_{\sigma_x}(t,z)-\Psi_{\sigma_x}(t,y)|=+\infty$ and we have again a contradiction, since by hypothesis $C$ is compact and forward invariant.
\end{proof}
It is easy to see that the singleton-forward invariant case occurs if and only if all the subsystems share the same equilibrium, or, in other words, the  switched affine system is simply a translation of a linear switched system. We have thus proved that, if the linearized system~\eqref{eq:DeAffinedSystem} is unstable and under the mild irreducibily assumption, non-trivial affine switched systems~\eqref{eq:Switchingsystems} have \emph{no} compact forward invariant sets.

\subsection{Outer Approximation of the Minimal Forward Invariant Set}
In this subsection we propose two numerical methods for computing an outer approximation of $\cK_\infty$, i.e., forward invariant sets $\cK\supseteq \cK_\infty$ as close as possible to $\cK_\infty$ (in a sense that we will clarify).
The first method is a direct consequence of Proposition~\ref{prop:Kvm} when considering \emph{quadratic norms}, providing forward invariant ellipsoids, and it is presented in the following statement.
	\begin{prop}\label{coro:lmi_arb}
		If there exists a symmetric matrix $S\in\R^{n\times n}$, a vector $c\in\R^n$ and a scalar $\kappa>0$ satisfying the inequalities
		\begin{subequations}\label{eq:lmis_arb}
        \begin{align}
		&SA_i^\top+A_iS\prec-2\kappa S,\;\;\;\;\;\forall i\in \M,\label{eq:LMI_arb_1} \\&\begin{bmatrix}
		\kappa^2 & (A_ic+b_i)^\top\\
		A_ic+b_i & S
		\end{bmatrix}\succ 0,\qquad \forall i\in\M,\label{eq:LMI_arb_2}
        \end{align}
		\end{subequations}
		then the ellipsoidal set
		\begin{equation}\label{eq:KQ_def}
		{\cK}_Q\coloneqq \{x\in\R^n~\vert~(x-c)^\top S^{-1}(x-c)\leq 1\}\supseteq \cK_\infty
		\end{equation}
		is a forward invariant set for system \eqref{eq:Switchingsystems} on $\cS$.
	\end{prop}
	\begin{proof}
	Consider system~\eqref{eq:Switchingsystems} translated to $c\in\R^n$, which is equivalent to replace $b_i\gets A_ic+b_i$ in~\eqref{eq:Switchingsystems}.
	From the term (2,2) in \eqref{eq:LMI_arb_2}, one has that $S^{-1}\succ0$. Inequality~\eqref{eq:LMI_arb_1}, in turn, implies that the norm $v(x) =\sqrt{x^\top S^{-1}x}$ satisfies the condition~\eqref{eq:LyapunovArbitraryNorm} in Lemma~\ref{lemma:ConverseLyap}. Moreover, the second inequality, by the Schur Complement Lemma (see~\cite[Section A.5.5]{boyd2004convex}), is equivalent to 
	\begin{align}
	    \kappa^2&> \max_{i\in\M} (A_ic+b_i)^\top S^{-1}(A_ic+b_i)= B_{max}^2
	\end{align}
	which ensures that $R=(\frac{B_{max}}{\kappa})^2<1$. Therefore, $x\notin {\cK}_Q$ implies $v(x)>1>R$ and therefore, by Proposition~\ref{prop:Kvm}, ${\cK}_Q$ is a forward invariant set.
	\end{proof}
An optimization problem to minimize the volume of the ellipsoidal set ${\cK}_Q$ can be stated considering the objective function $\min_{S,c,\kappa} \ln\det(S)$ subject to~\eqref{eq:lmis_arb}. The volume of ${\cK}_Q$ is proportional to $\det(S)$ and the $\ln$ function makes the objective function concave in $S$. However, for a given $\kappa>0$, this problem is a concave-minimization problem, which is generally hard to solve, \cite{Rockafellar70}. Some  strategies (see \cite{apkarian2000robust}) to handle this problem are convex-optimization methods  for local minimization and branch-and-bound algorithms for global minimization. A good alternative objective function, leading to a convex optimization problem is ${\rm Tr}(S)$, which, instead of minimizing the volume of the ellipsoid $\cK_Q$, minimizes the sum of the square lengths of the semi-axes, see \cite[Section~2.2.2]{boyd2004convex}. 
Notice that the point $c\in \R^n$ to which the system is translated is also a variable of the optimization problem, which allows us not only to estimate the size of $\cK_\infty$ but also to optimize a suitable center in the state space for the ellipsoid containing it.

 Nevertheless, it is well known that the existence of a common \emph{quadratic} Lyapunov function for the linearized part (i.e., a matrix $S>0$ satisfying~\eqref{eq:LMI_arb_1}) is a restrictive condition for stability. Moreover, Example~\ref{ex:FirstExample} shows how the set $\cK_\infty$ is in general non-convex; it thus seems natural, in order to improve our estimation of the minimal forward invariant set $\cK_\infty$, to consider sub-level sets of more general functions. In what follows, we propose a construction based on sub-level set of non-homogeneous  \emph{sum-of-squares (SOS) polynomials}, see~\cite{papachristodoulou2002construction} for details. 
\begin{prop}[Outer Approximaton via SOS polynomials]\label{prop:sos_arb}
If there exist a non-trivial polynomial $V(x)\in \R[x]$ of degree $d\in\N$ and scalars $r>0$, $\beta\geq0$ such that 
the following SOS constraints are satisfied
\begin{subequations}
\label{eq:sos_arb}
        \begin{align}
       V(x)-\epsilon\|x\|_{d}^{d}~\text{\small\rm is SOS} \label{eq:sos_arb_1}&\\
        -\left(\frac{\partial V}{\partial x}(x)\right)^{\top}(A_ix+b_i) -\beta(V(x)-r)~\text{\small\rm is SOS}&\quad\forall i\in\M\label{eq:sos_arb_2}
        \end{align}
\end{subequations}     
for some $\epsilon>0$, then the set 
\begin{equation}
     \cK_{\rm SOS}\coloneqq  \{x\in\R^n~\vert~V(x)\leq r\}\supseteq \cK_\infty\label{eq:Ksos_def}
\end{equation}
is a forward invariant set for system \eqref{eq:Switchingsystems} on $\cS$.
\end{prop}
\begin{proof}
The SOS constraint~\eqref{eq:sos_arb_1} ensures that the polynomial $V(x)\in \R[x]$ is positive definite whereas~\eqref{eq:sos_arb_2} implies 
\begin{equation}
    \left(\frac{\partial V}{\partial x}(x)\right)^{\top}(A_ix+b_i)<0,~~\forall i\in\M,~\forall x\notin \cK_{\rm SOS},
\end{equation}
recalling the definition of $\cK_{\rm SOS}$ in~\eqref{eq:Ksos_def}. This condition implies that $V(x)$ decreases along solutions of~\eqref{eq:Switchingsystems} that lie outside of $\cK_{\rm SOS}$. Therefore, this set is attractive and forward invariant for system \eqref{eq:Switchingsystems} undergoing arbitrary switching.
\end{proof}

\begin{rem}
Minimizing, or even computing, the volume of the set $\cK_{\rm SOS}$ defined in~\eqref{eq:Ksos_def} is, in general, a very intricate task, see~\cite{lasserre2019volume} for some discussions. Therefore, a simple way to optimize the over estimation of $\cK_\infty$ is to minimize $r>0$ subject to \eqref{eq:sos_arb}, where $\epsilon>0$ is chosen to avoid the trivial solution $V(x)=0$. {Also, this optimization problem can be efficiently solved by a bisection procedure over the scalar $\beta>0$, as the problem becomes convex whenever $\beta$ is given.}

Another important remark is that any polynomial $V(x)$ satisfying \eqref{eq:sos_arb_2} and  $V(x)-\epsilon\|x\|_{d}^{d}~\text{\small is SOS}$ can be decomposed as $V(x) = V_H(x) + V_R(x)$ where $V_H(x)$ is a homogeneous polynomial of the same degree $d\in\N$ as $V(x)$ and $V_R(x)$ are the remaining terms. Also, $w(x) = (V_H(x))^{\frac{1}{d}}$ can be shown to be a Lyapunov function for the linearized system~\eqref{eq:DeAffinedSystem}, satisfying the conditions \eqref{eq:Cond1} and \eqref{eq:Cond2} for some $\kappa>0$. Indeed, the SOS constraint~\eqref{eq:sos_arb_2} implies that $-\left(\frac{\partial V}{\partial x}(x)\right)^{\top}(A_ix+b_i)-\beta r\geq \beta V(x)$ for all $i\in\M$ which, in turn, implies that $-\left(\frac{\partial V_H}{\partial x}(x)\right)^{\top}A_ix\geq \beta V_H(x)$ for all $i\in\M$. Therefore, $w(x) = (V_H(x))^{\frac{1}{d}}$ is a common Lyapunov function for the linearized system~\eqref{eq:DeAffinedSystem}  {homogeneous of degree $1$. This in particular implies the existence of a norm $v:\R^n\to \R_{\geq 0}$ satisfying the conditions in Lemma~\ref{lemma:ConverseLyap} for $\kappa=\beta/d$. Thus, recalling Proposition~\ref{prop:KLstability}, the approximating technique given by Proposition~\ref{prop:sos_arb} provides also an upper bound on the decay rate of system~\eqref{eq:Switchingsystems}.} \hfill $\triangle$
\end{rem}
The following example illustrates both methods presented in this section.
\begin{figure}[t!]
    \centering
	\def\svgwidth{1.1\linewidth}
    \hspace*{-1.4cm}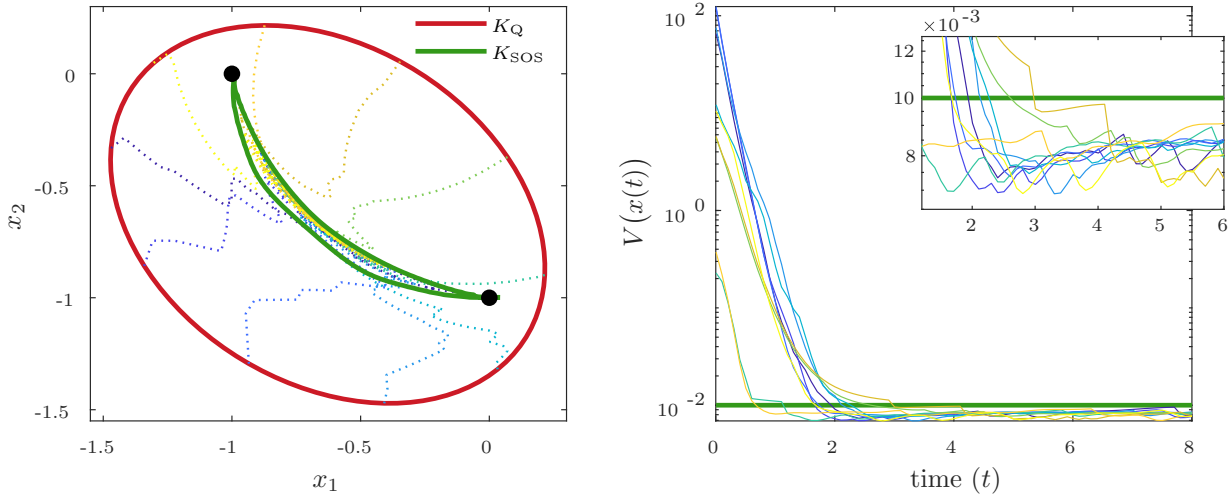

    \caption{Illustration of Example~\ref{ex:arb_swi_Kinf}. On the left, a representation of the state space with the forward invariant sets $K_{\rm Q}$ and $K_{\rm SOS}$, the equilibria of each subsystem (black dots), and $10$ trajectories undergoing random switching. On the right, the evaluation of the polynomial $V(x)$ associated to $K_{\rm SOS}$ along each of these trajectories and the horizontal line represents the level set used to define $K_{\rm SOS}$. } 
    \label{fig:arb_swi_Kinf}
\end{figure}
\begin{myexample}\label{ex:arb_swi_Kinf}
Consider the switched affine system introduced in Example~\ref{ex:FirstExample}. We constructed two forward invariant sets that serve as outer approximations of the minimal one $\cK_\infty$. These sets are respectively based on Proposition~\ref{coro:lmi_arb} and Proposition~\ref{prop:sos_arb}, and all related optimization problems were solved using \textsc{yalmip}~\cite{Lofberg2004} and \textsc{mosek}. The first one is the ellipsoidal region $\cK_{\rm Q}$, which is defined in \eqref{eq:KQ_def} and relies on the existence of a common quadratic Lyapunov Function for the linearized system~\eqref{eq:DeAffinedSystem}. For the first outer bound, the obtained matrix $S\succ0$ and center $c\in \R^n$ were
\[
S =\begin{bmatrix}
      	0.7120 & -0.2021 \\ 
      	 -0.2021 & 0.7120
     \end{bmatrix},~~c=\begin{bmatrix}
      	-0.6291 \\ 
      	 -0.6291
     \end{bmatrix},
\]
which solve the problem $\min_{S,c,\kappa} {\rm trace}(S)$ subject to \eqref{eq:lmis_arb} for $\kappa=0.4785$. The second set is  $\cK_{\rm SOS}$, which was obtained by solving the SOS optimization problem $\min_{V(x),r} r$ subject to \eqref{eq:sos_arb}, where $\epsilon=10^{-2}$ is chosen to avoid a trivial solution $V(x)\equiv0$. Restricting the search for the polynomial $V(x)$ to those with maximum total degree $d=12$ and choosing $\beta=1$, the convex optimization problem yielded the forward invariant set $\cK_{\rm SOS}$ as defined in~\eqref{eq:Ksos_def}. {The average elapsed-times for the solution of these optimization problems were $0.1096$ and $0.2128$ seconds, respectively, which were executed on an Intel$^{\circledR}$ Core$^{\text{\tiny TM}}$ i7-10610U CPU @ 1.80~GHz$\times$8 with $16$~GB of memory running Matlab R2020a on Ubuntu 20.04.} Both $\cK_{\rm Q}$ and $\cK_{\rm SOS}$  are represented in Figure~\ref{fig:arb_swi_Kinf} (left) along with the equilibria of each subsystem (black dots) and 10 trajectories starting on the boundary of $\cK_{\rm Q}$ that undergo random switching. The value of $V(x)$ evaluated along each of these trajectories is also depicted in Figure~\ref{fig:arb_swi_Kinf} (right) as a function of time, where the horizontal line represents the level set defining $\cK_{\rm SOS}$ with $r=0.0110$. In detail, we can notice that the trajectories were all attracted to $\cK_{\rm SOS}\supseteq \cK_\infty$ and never left it once in its interior, as expected.
\end{myexample}
The numerical experiments carried out in the previous example indicate that, although the SOS-based method given in~Proposition~\ref{prop:sos_arb} may yield more precise outer approximations for $\cK_\infty$ {(possibly non-convex), the LMI-based approach from~Proposition~\ref{coro:lmi_arb} can be computationally more attractive since it scales better with the system dimension. This concludes the stability study of arbitrarily switching signals and we now move on to present the dwell-time switching results.}


\section{Dwell-Time Switching Signals}\label{sec:DwellTime}
In many practical situations, stability under \emph{arbitrary} switching signals is a restrictive requirement, and system~\eqref{eq:Switchingsystems} is known to follow prescribed switching rules $\sigma\in \cS$ that satisfy constraints  bounding the frequency of the switching events. For that reason, in this section, given any threshold $\tau>0$ representing a \emph{dwell-time}, we restrict our analysis to the switching signals belonging to the \emph{dwell-time signal class}  $\cS_{\dw}(\tau)$ defined in~\eqref{eq:DwelltimeSignal}.
\subsection{Stability Analysis and Ultimate Boundedness}
Stability of switched linear systems as in~\eqref{eq:DeAffinedSystem} under dwell-time assumption is a well-studied problem and it is known that, if all the matrices $A_1,\dots,A_M\in \R^{n\times n}$ are Hurwitz, there exists a (large enough) dwell-time for which the switched system~\eqref{eq:DeAffinedSystem} is UGAS~\cite[Lemma~2]{Mor96},~\cite[Chapter 3]{liberzon}.  In this subsection, given $\tau>0$,  we analyze the stability/convergence properties of the switched \emph{affine} system~\eqref{eq:Switchingsystems} on  $\cS_{\dw}(\tau)$.  First of all we observe in the next result that, if the linearized system is unstable under arbitrary switching signals, there is no hope to find non-trivial forward invariant sets on $\cS_{\dw}(\tau)$, for \emph{any} dwell time $\tau>0$.
\begin{lemma}\label{lem:NonExistenceOFFIsetsDT}
Consider $\cF=\{ (A_i, b_i)_{i\in \M}\,\vert A_i\in \R^{n \times n}, b_i\in \R^n\}$, and suppose that the linearized system~\eqref{eq:DeAffinedSystem} is unstable on $\cS$. Then, \emph{for any} $\tau>0$ there does not exist a compact forward invariant set with non-empty interior on $\cS_{\dw}(\tau)$. If moreover the set $\cA=\{A_1,\dots, A_M\}$ is irreducible, there exists no compact non-singleton forward invariant set on $\cS_{\dw}(\tau)$.
\end{lemma}
\begin{proof}[Sketch of Proof]
The proof follows by Proposition~\ref{prop:Instability} and is briefly sketched here. Given any $\tau>0$, suppose by contradiction that $C\subset\R^n$ is a compact forward invariant set with non-empty interior for~\eqref{eq:Switchingsystems} on $\cS_{\dw}(\tau)$. Without loss of generality, by Lemma~\ref{lemma:convexity}, we can suppose that $C$ is convex. By Proposition~\ref{prop:Instability}, $C$ is not forward invariant on $\cS$, i.e., there exists an $x\in C$, a $\sigma\in \cS$ and a $T>0$ such that $\Psi_\sigma(T,x)\notin C$. Consider $\overline t\coloneqq \sup_{t\in [0,T]}\{\Psi_\sigma(t,x)\in C\}$ and denote $y=\Psi_\sigma(\overline t,x)$. It can be seen that $y\in \partial C$ and, by convexity of $C$ and recalling the Nagumo Theorem for differential inclusions~\cite[Chapter 4, pag. 175]{AubinCellina84}, there exists $i\in \M$ such that $A_iy+b_i\notin \cT_C(y)$, where $\cT_C(y)$ denotes the tangent cone to $C$ at $y$.
It is thus clear that considering the constant signal $\wt \sigma\in \cS_{\dw}(\tau)$ defined by $\wt \sigma(s)\equiv i$ for all $s\in \R_{\geq 0}$, we have that there exists $t_0>0$ such that $\Psi_\sigma(t_0,y)\notin C$, contradicting the forward invariance of $C$ for~\eqref{eq:Switchingsystems} on $\cS_{\dw}(\tau)$. The second part, assuming the irreducibility of $\cA$, follows from the reasoning already presented in the proof of Proposition~\ref{prop:Instability}.
\end{proof}
Despite this limiting result, UGAS of the linearized system~\eqref{eq:DeAffinedSystem} on $\cS_{\dw}(\tau)$ does imply some remarkable asymptotic properties of the switched affine system~\eqref{eq:Switchingsystems} on $\cS_{\dw}(\tau)$, and thus we need to recall the following (weak) stability/boundedness notion. 
\begin{defn}\label{defn:UGUB}
Given any class of switching signals $\wt \cS\subset \cS$, the switched affine system \eqref{eq:Switchingsystems} is said to be \emph{uniformly globally ultimately bounded (UGUB) on $\wt \cS$} if there exists a compact set $\cV\subset \R^n$  such that
\[
\begin{aligned}
\forall x\in \R^n,\;\forall \sigma\in \wt \cS,\; \exists\, T(\sigma,x)\geq 0\;\; \text{such that}\;\forall t\geq T(\sigma,x), \;\;\Psi_\sigma(t,x)\in \cV .
\end{aligned}
\]
In this case the compact set $\cV\subset\R^n$ is said to be a  \emph{uniform bounding region}.\hfill $\triangle$
\end{defn}
Note that the set $\cV$ in Definition~\ref{defn:UGUB}, in general, is \emph{not} forward invariant and thus,  not an attractor.
In the following we prove that, if for some $\tau>0$ the linearized system~\eqref{eq:DeAffinedSystem} is UGAS on $\cS_{\dw}(\tau)$ then the switched affine system~\eqref{eq:Switchingsystems} is UGUB.\\


\begin{thm}\label{thm:DwellTimeBoundedness}
For any given $\tau\in \R_{\geq 0}$ and  $\cF=\{ (A_i, b_i)_{i\in \M}\,\vert A_i\in \R^{n \times n}, b_i\in \R^n\}$, suppose that the linearized system~\eqref{eq:DeAffinedSystem} is UGAS on $\cS_{\dw}(\tau)$. Then the switched affine system \eqref{eq:Switchingsystems} is uniformly globally ultimately bounded on $\cS_{\dw}(\tau)$.
\end{thm}
For any $i$-th affine subsystem $\dot x=A_ix+b_i,~i\in\M$,  we denote by $\Psi_i(\cdot, x):\R_{\geq 0}\to\R^n$ the solution starting at $x(0)=x$ and by $x_{ei}\coloneqq -A_i^{-1}b_i$ the corresponding equilibrium.
The proof of Theorem~\ref{thm:DwellTimeBoundedness} requires the following preliminary result. 
\begin{lemma}\label{lemma:SecurityBall}
Consider the switched affine system given in~\eqref{eq:Switchingsystems}. For some $\tau>0$, suppose that the linearized system~\eqref{eq:DeAffinedSystem} is UGAS on $\cS_{\dw}(\tau)$. 
Then, there exist translated norms\footnote{A function $w:\R^n\to \R$ is said to be a \emph{translated norm} if there exist a norm $v:\R^n\to \R$ and a vector $c\in \R^n$ (called the \emph{center} of $w$) such that $w(x)=v(x-c)$, for all $x\in \R^n$.} $\wt v_i:\R^n\to \R$, a scalar $\wt\kappa>0$ and a collection of compact sets $\X_{i}\subset\R^n,~ i\in \M$, such that
   
    \begin{subequations}
     \begin{equation}
        x_{ei}\in \Inn(\X_i),\;\;\forall\; i\in \M,\label{eq:CondCenterAff}
    \end{equation}
    \begin{equation}
    \wt v_i\big(\Psi_i(t,x)\big)\leq \wt v_i(x),\;\;\;\,\forall x\in \R^n\setminus \Inn(\X_{i}),\,\,\forall t\in \R_{\geq 0},\forall i\in \M,\label{eq:Cond1Aff}
    \end{equation}
   \begin{equation}
       \wt v_i\big(\Psi_i(\tau,x)\big)\leq e^{-\wt \kappa\tau}\wt v_j(x),\;\;\;\forall x\in \R^n\setminus \Inn(\X_{j}),\;\forall (i,j)\in \M^2.\label{eq:Cond2Aff}
        \end{equation} \end{subequations}
      
\end{lemma}
\begin{proof}
Since the linearized system~\eqref{eq:DeAffinedSystem} is UGAS on $\cS_\dw(\tau)$, we consider the scalar $\kappa>0$ and norms $v_1,\dots, v_M:\R^n\to \R_{\geq 0}$ verifying the conditions \eqref{eq:Cond1} and  \eqref{eq:Cond2} defined in Lemma~\ref{lemma:DwellTimeConverseLyapLinear}. Define $\wt v_i(x)\coloneqq  v_i(x-c_i)$ with arbitrary $c_i\in\R^n$ and take any positive scalar $\wt\kappa<\kappa$.
First, to demonstrate that~\eqref{eq:Cond1} implies \eqref{eq:Cond1Aff}, consider any $i\in \M$, any $x\in \R^n$, and any $t\in \R_{\geq 0}$. 
Thus, computing
\begin{align}
 \wt v_i\big(\Psi_i(t,x)\big)\hspace{-0.05cm}&= v_i\left(e^{A_it}x+\int_0^t e^{A_i(t-s)}b_i\,ds-c_i + \Psi_i(t,c_i)- \Psi_i(t,c_i) \right)\nonumber\\&=
v_i\left(e^{A_it}x+\int_0^t e^{A_i(t-s)}b_ids-e^{A_it}c_i-\hspace{-0.05cm}\int_0^t e^{A_i(t-s)}b_ids +  \Psi_i(t,c_i)-c_i\right)\nonumber\\&\leq
v_i\left(e^{A_it}(x-c_i)\right) +v_i( \Psi_i(t,c_i)-c_i)\nonumber\\&
\leq e^{-\kappa t}v_i(x-c_i)+v_i( \Psi_i(t,c_i)-c_i)\nonumber\\&=
e^{-\kappa t}\wt v_i(x)+v_i( \Psi_i(t,c_i)-c_i).\label{eq:TranslatedNorms}
\end{align}
Also, notice that
\begin{align}
     \Psi_i(t,c_i)-c_i &=e^{A_it}c_i+\int_0^t e^{A_i(t-s)}b_i\,ds-c_i\nonumber\\
     &=(e^{A_it}-I)c_i+\int_0^t e^{A_i(t-s)}b_i\,ds\nonumber\\
     &=\int_0^t e^{A_i(t-s)}A_ic_i\,ds+\int_0^t e^{A_i(t-s)}b_i\,ds\nonumber\\
     &=\int_0^t e^{A_i(t-s)}(A_ic_i+b_i)\,ds,
\end{align}
which, combined with \eqref{eq:TranslatedNorms} yields
\begin{align}
    \wt v_i\big(\Psi_i(t,x)\big) &\leq e^{-\kappa t}\wt v_i(x)+v_i\left(\int_0^t e^{A_i(t-s)}(A_ic_i+b_i)\,ds\right)\nonumber\\
    &\leq e^{-\kappa t}\wt v_i(x)+\int_0^t v_i\left( e^{A_i(t-s)}(A_ic_i+b_i)\right)\,ds\nonumber\\
    &\leq e^{-\kappa t}\wt v_i(x)+\int_0^t e^{-\kappa (t-s)} v_i\left( A_ic_i+b_i\right)\,ds\nonumber\\
    &\leq e^{-\kappa t}\wt v_i(x)+ \frac{1- e^{-\kappa t}}{\kappa} v_i\left( A_ic_i+b_i\right),\label{eq:BoundDwellFlow}
\end{align}
which ensures \eqref{eq:Cond1Aff} for all $x\in\R^n$ such that $\wt v_i(x)\geq \kappa^{-1}v_i(A_ic_i+b_i)$. Moreover we note that~\eqref{eq:BoundDwellFlow} holds globally if one chooses $c_i=x_{ei}$ but, for the sake of generality (and for numerical reasons illustrated in what follows), this demonstration is carried out for arbitrary centers $c_i,~i\in\M$. 

 To prove~\eqref{eq:Cond2Aff}, consider any $x\in \R^n$ and any $(i,j)\in \M^2$. Reasoning similarly to~\eqref{eq:TranslatedNorms}, we have

\begin{align}
\wt v_i\big(\Psi_i(\tau,x)\big)&\leq
v_i\left(e^{A_i\tau}(x-c_i)\right) +v_i( \Psi_i(\tau,c_i)-c_i)\nonumber\\&\leq
 e^{-\kappa\tau}v_j(x-c_i) +\frac{1- e^{-\kappa \tau}}{\kappa} v_i\left( A_ic_i+b_i\right) \nonumber\\&\leq 
 e^{-\kappa\tau} \big(v_j(x-c_j)+v_j(c_j-c_i)\big)+\frac{1- e^{-\kappa \tau}}{\kappa} v_i\left( A_ic_i+b_i\right)\nonumber\\&=
 e^{-\kappa\tau} \big(\wt v_j(x)+\wt v_j(c_i)\big)+\frac{1- e^{-\kappa \tau}}{\kappa} v_i\left( A_ic_i+b_i\right).
\end{align}
Notice that, for any $0<\wt \kappa<\kappa$, we have that 
\[
\wt v_i\big(\Psi_i(\tau,x)\big)\leq e^{-\wt \kappa \tau}\wt v_j(x),
\]
if $x\in\R^n$ is such that 
\begin{equation}\label{eq:DefnR_ij}
    \wt v_j(x)\geq \frac{e^{-\kappa \tau}\wt v_j(c_i) +\kappa^{-1}(1- e^{-\kappa \tau})v_i\left( A_ic_i+b_i\right)}{e^{-\wt\kappa \tau}-e^{-\kappa \tau}}=:R_{ij}.
\end{equation}
Hence, conditions in~\eqref{eq:Cond2Aff} are verified considering
\begin{equation}\label{eq:Vij_def}
    \X_{i}\coloneqq \{x\in \R^n\,\vert\,\wt v_i(x)\leq R_\X\},
\end{equation}
with $R_\X \coloneqq  \max_{(i,j)\in\M^2} R_{ij}$. Notice that~\eqref{eq:Cond1Aff} also holds for the same $\X_i$ because
\begin{align}\label{eq:Bounds}
    R_{ij}&=\frac{e^{-\kappa \tau}\wt v_j(c_i) +\kappa^{-1}(1- e^{-\kappa \tau})v_i\left( A_ic_i+b_i\right)}{e^{-\wt\kappa \tau}-e^{-\kappa \tau}}\nonumber\\
    &\geq\frac{\kappa^{-1}(1- e^{-\kappa \tau})v_i\left( A_ic_i+b_i\right)}{e^{-\wt\kappa \tau}-e^{-\kappa \tau}}\nonumber\\
    &> \frac{\kappa^{-1}(1- e^{-\kappa \tau})v_i\left( A_ic_i+b_i\right)}{1-e^{-\kappa \tau}}\nonumber\\
    &= \kappa^{-1}v_i\left( A_ic_i+b_i\right).
\end{align}
{Finally, considering $x=x_{ei}$ in~\eqref{eq:BoundDwellFlow} and any $t>0$, we note that $\wt v_i(x_{ei})\leq \kappa^{-1}v_i(A_ic_i+b_i)$. This, by~\eqref{eq:Bounds} implies that $\wt v_i(x_{ei})<R_{ij}$ for every $(i,j)\in \M^2$, and thus condition~\eqref{eq:CondCenterAff} holds, concluding the proof.}
\end{proof}


We can now prove Theorem~\ref{thm:DwellTimeBoundedness}.
\begin{proof}[Proof of Theorem~\ref{thm:DwellTimeBoundedness}]
Consider translated norms $\wt v_1,\dots, \wt v_M:\R^n\to \R_{\geq 0}$, a scalar $\wt \kappa>0$ and sets $\X_{i}$ defined in~\eqref{eq:Vij_def} satisfying the conditions of Lemma~\ref{lemma:SecurityBall}.
Define the compact set $\X\coloneqq \bigcup_{i\in\M} \X_{i}$ and take 
\begin{equation}\label{eq:def_V}
\begin{aligned}
\V &\coloneqq  \bigcup_{i\in\M} \V_i, \text{ with}\\
{\V}_i &\coloneqq  \bigcup_{j\in\M} \bigcup_{t\in[0,\tau]} \Psi_j(t,\X_i)
\end{aligned}
\end{equation}
where, for every $j \in \M$ and $t\in \R_{\geq 0}$, the set-valued map $\Psi_j(t,\cdot):\R^n\rightrightarrows\R^n$ is defined by $\Psi_j(t,S)\coloneqq \bigcup_{x\in S}\left\{\Psi_j(t,x)\right\}$ for any set $S\subseteq\R^n$.
An illustration of the sets $\V$ and $\X$ is shown in Figure~\ref{fig:setsXV}.
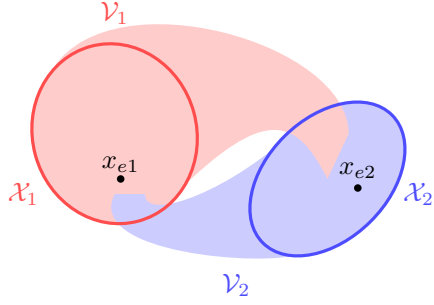
\begin{figure}
    \centering
    \def\myred{red!70}
    \def\myblue{blue!70}
    \def\mylred{red!20}
    \def\mylblue{blue!20}
    
    	\begin{tikzpicture}[scale=0.4]
	    \coordinate (xe1) at (0,0);
	    \coordinate (xe2) at (7,-1.5);	
	    \fill[rotate=15,very thick,color=\mylred] (xe1) ellipse (2.7cm and 3cm);
		\fill[rotate=45,very thick,color=\mylblue] (xe2) ellipse (3cm and 2cm);
    	
    	\fill[color=\mylblue] (6,0.2) .. controls +(-2,-2) and +(0,-1) .. (1,-2) -- %
    	                    (0,-2) .. controls +(-1,-2) and +(-1,-.4) .. (6,-4) -- (6,0.2);
        \fill[color=\mylred] (-2,2.2) .. controls +(2,3) and +(0,3) .. (7.7,0) -- %
        (7,-1.5) .. controls +(-1,2) and +(3,3) .. (2,-2) -- (-2,2.2);
        
    	\draw[rotate=15,very thick,color=\myred] (xe1) ellipse (2.7cm and 3cm);
		\draw[rotate=45,very thick,color=\myblue] (xe2) ellipse (3cm and 2cm);
	%
		
		\node at (0,4) {\color{\myred}$\mathcal{V}_1$};
		\node at (-3,-2) {\color{\myred}$\mathcal{X}_1$};
		\node at (10,-2) {\color{\myblue}$\mathcal{X}_2$};
		\node at (4,-5) {\color{\myblue}$\mathcal{V}_2$};
\node[circle, fill, scale=0.3] at (0.2,-1.5) {};
\node at (0.2,-1.0) {$x_{e1}$};
\node[circle, fill, scale=0.3] at (8,-1.8) {};
\node at (8,-1.2) {$x_{e2}$};
	\end{tikzpicture}
    \caption{For the case of $2$-mode switched affine system, a qualitative representation of the sets $\X=\X_1 \cup \X_2$ (solid lines) and $\V=\V_1 \cup \V_2$ (filled area). }
    \label{fig:setsXV}
\end{figure}
Then, we establish uniform global ultimate boundedness proving the following two claims:
\begin{itemize}[leftmargin=1.5cm]
    \item[Claim 1.] For any $x\notin  \X$ and any $\sigma\in \cS_{\dw}(\tau)$ there exists a finite switching instant $t_k^\sigma>0$ such that $\Psi_\sigma(t^\sigma_k,x)\in  \X$.
    \item[Claim 2.] For any $x\in   \X$ and any $\sigma \in \cS_{\dw}(\tau)$, $\Psi_\sigma(t,x)\in  \V$, for all $t\in \R_{\geq 0}$.
\end{itemize}
Claims 1 and 2 ensure the UGUB property in Definition~\ref{defn:UGUB}, with the compact set $\cV$ defined in~\eqref{eq:def_V}.
Indeed, we note that we can suppose, without loss of generality, that $\sigma$ has infinitely many discontinuities: if $\sigma$ is eventually constant then, considering any $x\in \R^n$, it holds that $\Psi_\sigma(t,x)\to x_{ei}$ as $t\to \infty$ for some $i\in \M$, and $x_{ei}\in \Inn(\X_i)\subseteq \X\subseteq \cV$  by Lemma~\ref{lemma:SecurityBall}, and thus there exists a $T(\sigma, x)\geq0$ as in Definition~\ref{defn:UGUB}.\\
To prove Claim 1, consider any $x\notin \X$ and any $\sigma \in \cS_{\dw}(\tau)$.
Define the instant $\hat t^\sigma\coloneqq \min\{t^\sigma_k\,\,\,\,\vert\,\Psi_\sigma(t^\sigma_k,x)\in \X\}$, i.e., the first switching instant for which $\Psi_\sigma(t,x)$ is inside $\X$, and let us show that $\hat t^\sigma$ is always finite. Consider any switching instant $t_k^\sigma<\hat t^\sigma$, and suppose, without loss of generality that $\sigma(t_k^\sigma)=i$. Let us denote, by simplicity $x(k)\coloneqq \Psi_\sigma(t^\sigma_{k}, x)$ and $T_k\coloneqq  t^\sigma_{k+1}-t^\sigma_k\geq \tau$. From the definition of $\hat t^\sigma$ we have  $x(k)\notin \X\supseteq \X_{j}$ for any $j\in\M$. Therefore, using \eqref{eq:Cond1Aff} and \eqref{eq:Cond2Aff} we obtain
\begin{equation}\label{eq:MainStpeProofofBoundednees}
\begin{aligned}
\wt v_i (x(k+1))&=\wt v_i\big(\Psi_i(T_k,x(k))\big)\\
&\leq \wt v_i\big(\Psi_i(\tau,x(k))\big)\\
&\leq e^{-\wt \kappa \tau}\wt v_j(x(k)),
\end{aligned}
\end{equation}
for any $j\in \M$. Iterating this property backward yields
\[
\wt v_i (x(k+1))\leq e^{-\wt \kappa k\tau }\wt v_j(x),
\]
for any $j\in \M$.
Since this holds for an arbitrary $t^\sigma_k<\hat t^\sigma$ and $e^{-\wt \kappa k\tau}\to 0$ as $k\to \infty$, we can conclude that $\hat t_\sigma$ is finite as, otherwise, a contradiction would take place.

To prove Claim 2, consider $x\in \X$, which implies that there exists $j\in\M$ such that $\wt v_j(x)\leq  R_\X$, from the definition in~\eqref{eq:Vij_def}. Notice that, given the definition of $\V$, we have $\X_i\subseteq\V_i$ for all $i\in\M$. Suppose $\sigma\in \cS_{\dw}(\tau)$ is such that $\sigma(0)=i\in\M$ and take any $t\in [0,\tau]$, which ensures $\sigma(t)=i$ given the dwell-time constraint. From the definition of $\V_i$, it is straightforward to verify that $\Psi_i(t,\X_j)\subseteq\V_j$ for all $t\in[0,\tau]$. Then, we can analyse the worst case with respect to $x\in\X_j$ as
\begin{align}
\wt v_i\big(\Psi_i(\tau,x)\big)& \leq \max_{x\in \X_j}\wt v_i\big(\Psi_i(\tau,x)\big)\nonumber\\
&= \max_{x\in \partial\X_j}\wt v_i\big(\Psi_i(\tau,x)\big)\nonumber\\
&\leq \max_{x\in \partial\X_j} e^{-\wt\kappa \tau} \wt v_j\big(x\big)\nonumber\\
&\leq  R_\X,\label{eq:invariance_Xi}
\end{align}
where the equality holds because the maximum of the convex function $\wt v_i\big(\Psi_i(\tau,\,\cdot\,)\big):\R^n\to \R$ inside the convex set $\X_j$ occurs on its boundary $\partial\X_j$ (see~\cite[Corollary~32.3.2]{Rockafellar70}) and the subsequent inequality follows from~\eqref{eq:Cond2Aff}. This implies that $\Psi_i(\tau,\X_j)\subseteq \X_i$ for all $(i,j)\in\M^2$ and thus, by~\eqref{eq:Cond1Aff}, $\Psi_i(t,x)\in\X_i$ for all $t\in[\tau,t^\sigma_1]$. We can iterate this argument for every $k\in \N$ such that $\Psi_\sigma(t^\sigma_k,x)\in \X$, and the proof is concluded.
\end{proof}

As proven in~Lemma~\ref{lem:NonExistenceOFFIsetsDT}, in general, for any $\tau>0$, non-trivial forward invariant sets for~\eqref{eq:Switchingsystems} on the class $\cS_\dw(\tau)$ do not exist, even if the linearized system~\eqref{eq:DeAffinedSystem} is UGAS on $\cS_\dw(\tau)$. For that reason, in Theorem~\ref{thm:DwellTimeBoundedness} we focused on the weaker \emph{uniform global ultimate boundedness} property, constructing a bounding region $\cV$, starting from a family of ``safety'' sets $\X_1,\dots, \X_M$. On the other hand, we can focus on a more relaxed notion of forward invariance, which is only concerned by the solutions \emph{evaluated along the discrete sequence of switching points}, as formally defined below.
\begin{defn}
 Given $\wt \cS\subset \cS$, we say that a compact set $C\subset \R^n$ is \emph{forward invariant} for~\eqref{eq:Switchingsystems} on $\wt \cS$ \emph{with respect to the switching points} if, for all $x\in C$, all $\sigma\in \wt \cS$, we have that
\[
\Psi_\sigma(t^\sigma_k,x)\in C, \;\;\forall x\in C,\;\;\forall \sigma\in \wt \cS, \;\;\forall t^\sigma_k\geq 0,
\]
where, we recall, $\{t_k^\sigma\}$ denotes the (finite or countable) set of discontinuities of the signal $\sigma\in \wt \cS$.
\end{defn}
\begin{prop}[$\cX$ is forward invariant w.r.t. switching instants] 
Under the hypothesis of Theorem~\ref{thm:DwellTimeBoundedness}, the set $\cX=\bigcup_{i\in \M}\cX_i$  (where the sublevel sets $\cX_i$ are defined in~\eqref{eq:Vij_def}) is forward invariant for~\eqref{eq:Switchingsystems} on $\cS_\dw(\tau)$ w.r.t. the switching points. 
\end{prop}
\begin{proof}
First of all, we note that, by~\eqref{eq:Cond1Aff}, for any $i\in \M$, the set $\cX_i$ is forward invariant for the subsystem $\dot x=A_ix+b_i$.
Let us now consider any $x\in \cX$ and any $\sigma\in \cS_\dw(\tau)$. By definition of $\cX$ there exists a $j\in \M$ such that $x\in \cX_j=\{x\in\R^n\;\vert\; \wt v_j(x)\leq R_\cX\}$.  Suppose $\sigma(0)=i_0\in \M$ and consider $t^\sigma_1\geq\tau$ as the first switching instant of $\sigma$. By forward invariance of $\cX_{i_0}$ with respect to the $i_0$-subdynamics, if $\Psi_\sigma(\tau,x)\in \cX_{i_0}$ then $\Psi_\sigma(t^\sigma_1,x)\in \cX_{i_0}$. By~\eqref{eq:invariance_Xi} we have
\[
\wt v_{i_0}(\Psi_{\sigma}(\tau,x))=\wt v_{i_0}(\Psi_{i_0}(\tau,x))\leq \wt v_j(x)\leq R_\cX
\]
proving that $\Psi_{i_0}(t^\sigma_1,x)\in \cX_{i_0}\subset \cX$. Summarizing, we have proven that,  for any $x\in \cX$ and any $\sigma\in \cS_\dw(\tau)$ we have $\Psi_\sigma(t^\sigma_1,x)\in \cX$; iterating the argument we conclude that $\cX$ is forward invariant with respect to the switching instants.
\end{proof}
We also note that if the set $\{t^\sigma_k\}$ is finite (corresponding to an eventually constant switching signal $\sigma$), we have that $\lim_{t\to +\infty}\Psi_\sigma(t,x)=x_{ei}\in \Inn(\cX_i)\subseteq \Inn(\cX)$, for some $i\in \M$. This proves that, for any eventually constant signal $\sigma\in \cS_\dw(\tau)$, there exists a finite $\overline T(\sigma,x)\geq0$ such that $\Psi_\sigma(t,x)\in \cX$ for all $t\in [\overline T(\sigma,x),+\infty)$. On the other hand, the solutions of~\eqref{eq:Switchingsystems} can possibly escape $\cX$ between switching times, and for that reason we introduced the ``security'' set $\cV$, which characterizes the UGUB property. For a graphical illustration, see Figure~\ref{fig:setsXV}.

Concluding this subsection, we discuss the relations between the dwell-time parameter $\tau>0$, the corresponding decay rate and the sets $\cX_i$ and $\cV$.
\begin{rem}[Dependence of $\cX_i$ and $\cV$ with respect to $\tau$]\label{rem:rho_dt_sets}
Consider the system $\cF=\{ (A_i, b_i)_{i\in \M}\,\vert A_i\in \R^{n \times n}, b_i\in \R^n\}$, and denote $\cA=\{A_1,\dots, A_M\}$ the set of the matrices describing the linear part of the subdynamics.
We denote the \emph{minimum dwell-time} by $\tau_{\min}(\cA)$, formally defined as
\[
\tau_{\min}(\cA)\coloneqq \inf\left\{\tau\geq 0\;\vert\;\eqref{eq:DeAffinedSystem} \text{ is GUAS on }\cS_\dw(\tau) \right\}.
\]
By Lemma~\ref{lemma:DwellTimeConverseLyapLinear}, an equivalent definition is
\[
\begin{aligned}
\tau_{\min}(\cA)=&\inf_{\tau\geq 0} \;\tau,~~~\text{s.t.} \\
&\exists\kappa>0\; \text{and norms } \;v_1,\dots, v_M:\R^n\to \R_{\geq 0} \;\text{such that \eqref{eq:Cond1} and~\eqref{eq:Cond2} hold. } 
\end{aligned}
\]
We recall, once again, that if $A_1,\dots, A_M$ are Hurwitz, then $\tau_{\min}(\cA)$ is finite, see~\cite{Mor96}. 
Now, given $\tau>\tau_{\min}(\cA)$ we define the \emph{decay rate} with respect to $\tau>0$, i.e.,
\begin{equation}\label{eq:DefnDecayDwellTime}
\kappa_\cA(\tau)\!\coloneqq\! \sup \left \{\kappa\geq 0\;\vert\;\exists \text{norms }v_1,\dots,v_M:\R^n\to \R_{\geq 0}\;\text{such that~\eqref{eq:Cond1} and \eqref{eq:Cond2} hold}\right \}.
\end{equation}
The computation/approximation of the minimum dwell-time is a well-studied problem in the literature, see for example~\cite{geromel2006stability,Bri15,CheCol12,BlaCol10,YuanLv20,DelPasAng22} and references therein. As showed in~\cite{Wir05} the function $\kappa_\cA:(\tau_{\min}(\cA),\infty)\to \R_{\geq 0}$ is continuous and non-decreasing, and it can be seen that $\lim_{\tau\searrow \tau_{\min}(\cA)} \,\kappa(\tau)=0$, and thus we extent $\kappa_\cA $ by continuity defining $\kappa_\cA(\tau_{\min}(\cA))=0$. Moreover, the function $\kappa_\cA$ is bounded from above, by the number $\overline \kappa\coloneqq \min_{i\in \M} \kappa(A_i)$, where, given a matrix $A\in \R^{n\times n}$, $\kappa(A)\coloneqq \max_{i\in \{1,\dots, n\}}\{\text{Re}(\lambda_i(A))\}$. As discussed in~\cite{Wir05}, if the set $\cA$ is irreducible, then the $\sup$ in \eqref{eq:DefnDecayDwellTime} can be replaced by a $\max$, i.e., for any $\tau\geq \tau_{\min}(\cA)$ we can find norms satisfying~\eqref{eq:Cond1} and \eqref{eq:Cond2} for $\kappa_\cA(\tau)$. In what follows we will assume that $\cA$ is irreducible, without loss of generality. Moreover, we assume that $\tau_{\min}(\cA)>0$ since the case $\tau_{\min}(\cA)=0$ correspond to the case of stability under arbitrary switching, already studied in Section~\ref{sec:Arbitrary}.

We want now to use these notions to analyze how the sets $\cX_i$ and $\cV$ are affected while varying the dwell-time parameter $\tau$. We organize our analysis in two cases, described in what follows:
\begin{itemize}[leftmargin=*]
    \item The case of \emph{slow switching signals}, i.e., when the dwell-time $\tau$ becomes arbitrarily large.
    \item The case of \emph{fast switching signals}, i.e., the case $\tau \searrow \tau_{\min}(\cA)$ in which the dwell-time threshold approaches the minimum dwell-time $\tau_{\min}(\cA)$.
\end{itemize} 
\emph{Slow Switching:} Consider any $\underline \tau>\tau_{\min}(\cA)$ and consider any $\kappa>0$ and $v_1,\dots, v_M:\R^n\to \R_{\geq 0}$ satisfying the condition in Lemma~\ref{lemma:DwellTimeConverseLyapLinear}. It can be verified that the same $\kappa$ and norms $v_1,\dots, v_M$ satisfy the conditions in Lemma~\ref{lemma:DwellTimeConverseLyapLinear} \emph{for any} $\tau\geq\underline \tau$. For such $v_1,\dots, v_M:\R^n\to \R$, we apply the construction in the proof of Lemma~\ref{lemma:SecurityBall} and  we choose $c_i=x_{ei}$ for any $i\in \M$. Given any $0<\wt \kappa <\kappa$,  from \eqref{eq:DefnR_ij} it can be seen that, for all $(i,j)\in \M^2$
\[
\lim_{\tau\to +\infty} R_{ij}(\tau)=\lim_{\tau\to+\infty}\frac{e^{-\kappa \tau}\wt v_j(x_{ei})}{e^{-\wt\kappa \tau}-e^{-\kappa \tau}}=\lim_{\tau\to+\infty}\frac{ \wt v_j(x_{ei})}{e^{(\kappa-\wt \kappa)\tau}-1}=0.
\]
This implies that $R_\cX(\tau)=\max_{(i,j)\in \M^2} R_{ij}(\tau)\to 0$ as $\tau\to +\infty$, and thus, considering the compact sets $\cX_i(\tau)$ defined in~\eqref{eq:Vij_def} (making explicit the dependence with respect to $\tau$) we have
\[
\cX_i(\tau)\to \{x_{ei}\},\;\text{as } \tau\to +\infty, \;\;\forall\;i\in \M,
\]
where the convergence is in the Hausdorff sense. Similarly, the set $\cV(\tau)$ defined in~\eqref{eq:def_V} converges to the union of a finite number of trajectories, i.e., to the set
\[
\bigcup_{(i,j)\in \M^2}\bigcup_{t\in \R_{\geq 0}}\{\Psi_j(t, x_{ei})\}.
\]
\emph{Fast Switching:} We study the case $\tau\searrow \tau_{\min}(\cA)$, showing that the sets $\cX_i$ and $\cV$ explode, i.e., are eventually unbounded. This is somewhat expected since we are approaching the instability limit $\tau_{\min}(\cA)$  of the linearized system~\eqref{eq:DeAffinedSystem}. To this aim, consider any decreasing sequence $\tau_k\searrow \tau_{\min}(\cA)$ and, for any $k\in \N$, consider $\kappa_k\coloneqq \kappa_\cA(\tau_k)$ and $v^k_1,\dots, v^k_M:\R^n\to \R$ any norms satisfying~\eqref{eq:Cond1} and~\eqref{eq:Cond2} for $\tau_k$ and $\kappa_k$. As discussed, it holds that $\kappa_k\searrow 0$ as $k\to \infty$. In following the construction of proof of Lemma~\ref{lemma:SecurityBall} we set $\wt \kappa_k\coloneqq \kappa_k/2$ and this is without loss of generality, since $\kappa_k\to 0$. Similarly, we set $c_i^k=x_{ei}$, for any $k\in \N$; the general case follows the same argument. Computing, recalling~\eqref{eq:DefnR_ij}, for any $(i,j)\in \M^2$ we have
\[
\lim_{\tau\searrow \tau_{\min}(\cA)} R_{ij}(\tau)=\lim_{k\to+\infty}\frac{e^{-\kappa_k \tau_k}\wt v^k_j(x_{ei})}{e^{-\frac{\kappa_k}{2} \tau_k}-e^{-\kappa_k \tau_k}}=\lim_{k\to+\infty}\frac{ \wt v^k_j(x_{ei})}{e^{\frac{\kappa_k}{2}\tau_k}-1}.
\]
Since, by hypothesis $\kappa_k\searrow 0$ and $\tau_k\searrow \tau_{\min}(\cA)>0$ as $k\to \infty$, we have that $\lim_{\tau\searrow \tau_{\min}(\cA)} R_{ij}(\tau)=+\infty$. This means that, while approaching the stability margin given by $\tau_{\min}(\cA)$, our outer approximations of the safety sets $\cX_i$, $i\in \M$ are growing unboundedly, providing no information on the asymptotic behavior of system~\eqref{eq:Switchingsystems}. Of course, the same exploding behavior is inherited by the bounding region $\cV$. The numerical examples in the next section illustrate this phenomenon.  \hfill $\triangle$
\end{rem}
\subsection{Numerical Construction of Bounding Regions}
To devise a numerical procedure allowing the study of the switched affine system~\eqref{eq:Switchingsystems} under dwell-time switching we present an LMI-based method relying on sufficient conditions for obtaining functions $\wt{v}_i,~i\in\M$ satisfying the conditions~\eqref{eq:Cond1Aff}-\eqref{eq:Cond2Aff} in Lemma~\eqref{lemma:SecurityBall}.
For this, we consider norms induced by quadratic forms as
\begin{equation}
    v_i(x) = \sqrt{x^{\top}P_ix}\label{eq:quad_norms_i}
\end{equation}
where $P_i,~i\in\M,$ are positive definite matrices. From these norms, we define the functions $\wt{v}_i(x)=v_i(x-c_i)$ for given centers $c_i\in\R^n,~i\in\M$. The following corollary is the core of the numerical procedure to construct the bounding region $\V$.

\begin{cor}
For given $\tau\in \R_{\geq 0}$ and  $\cF=\{ (A_i, b_i)_{i\in \M}\,\vert A_i\in \R^{n \times n}, b_i\in \R^n\}$, the switched affine system \eqref{eq:Switchingsystems} is uniformly globally ultimately bounded on $\cS_{\dw}(\tau)$ if there exist positive definite matrices $P_i,W_{ij}\in\R^{n\times n}$ and vectors $c_i,d_{ij}\in\R^n$ satisfying the inequalities
\begin{subequations}
 \begin{align}
		\A_i^\top\P_i+\P_i\A_i&\prec-\E_{ii}\label{eq:lmi_flow}\qquad\forall i\in\M\\
		e^{\A_i^\top\tau}\P_i e^{\A_i\tau}-\P_j &\prec-\E_{ij}\qquad\forall (i,j)\in\M^2, i\neq j\label{eq:lmi_jump}
\end{align}
\end{subequations}
with 	
\begin{equation}
	\P_i = \begin{bmatrix}
	P_i & -P_ic_i\\
	-c_i^{\top}P_i & c_i^{\top}P_ic_i
	\end{bmatrix},~ \A_i =\begin{bmatrix}
		A_i & b_i \\
		0   &  0
		\end{bmatrix},~\E_{ij}=\begin{bmatrix}
	W_{ij} & -W_{ij}d_{ij}\\
	-d_{ij}^{\top}W_{ij} &  d_{ij}^{\top}W_{ij}d_{ij} -1
	\end{bmatrix}.\label{eq:def_cals}
\end{equation}
\end{cor}
\begin{proof}
Let us show that \eqref{eq:lmi_flow}-\eqref{eq:lmi_jump} imply \eqref{eq:Cond1Aff}-\eqref{eq:Cond2Aff}, respectively. First, notice that, using the augmented state vector $\xi=[x^\top~1]^\top$ one can rewrite the system~\eqref{eq:Switchingsystems} as $\dot{\xi} = \A_\sigma \xi$ and the functions $\wt v_i(x)=\sqrt{\xi^\top\P_i\xi}$. Consequentially, the solution to this system is simply given by $\Psi_i(t,x)=e^{\A_i t} \xi$. Also, notice that the set $\Eb_{ij} \coloneqq  \{x\in\R^n~\vert~\xi^\top \E_{ij}\xi\leq 0 \}$ defines a generic ellipsoid in $\R^n$, see~\cite[Section~3.7]{boyd1994linear}. For an arbitrary $x\in\R^n$, multiplying~\eqref{eq:lmi_flow} to the left by $\xi^\top$ and to the right by $\xi$ implies that
\[
\frac{d}{dt}\wt v_i(\Psi_i(t,x))^2 < - \xi^\top\E_{ii}\xi,\;\;\forall t\in \R_{\geq 0},
\]
which readily implies that \eqref{eq:Cond1Aff} holds for some level set $\X_i$ of $\wt v_i(x)$ containing the ellipsoid $\Eb_{ii}$. Analogously, multiplying~\eqref{eq:lmi_jump} to the left by $\xi^\top$ and to the right by $\xi$ implies that
\[
\wt v_i(\Psi_i(\tau,x))^2 - \wt v_j(x)^2 <-\xi^\top \E_{ij}\xi,
\]
which, in turn, implies that \eqref{eq:Cond1Aff} also holds for some level set $\X_i$ of $\wt v_i(x)$ containing $\Eb_{ij}$ with some positive scalar $\tilde{\kappa}$, as strict inequalities are considered. Therefore $\X_i$ can be defined as a level set of $\wt v_i(x)$ containing $\Eb_{ij}$.  The proof that $x_{ei}\in\Inn(\X_i)$ is done by multiplying~\eqref{eq:lmi_flow} to the left by $\xi_{ei}^\top = [x_{ei}^{\top}~1]$ and to the right by $\xi_{ei}$ and noticing that $\A_i\xi_{ei}=0$, which leads to $\xi_{ei}^\top\E_{ii}\xi_{ei}<0$, concluding the proof.
\end{proof}

\begin{rem}[Equivalence with the conditions of~\cite{geromel2006stability}] \label{rem:WeAREstrongas,JOse}
We highlight that the conditions \eqref{eq:lmi_flow}-\eqref{eq:lmi_jump} are equivalent, in terms of conservatism, to the classical dwell-time stability conditions presented in~\cite[Theorem~1]{geromel2006stability} for the linearized system~\eqref{eq:DeAffinedSystem}. In this result, the sufficient condition for stability is given by  the existence of matrices $\wt P_i\in\R^{n\times n}$, $\wt P_i\succ0$ such that
\begin{subequations}
  \begin{align}
		A_i^\top\wt P_i+\wt P_iA_i&\prec0\label{eq:lmi_flow_gero}\qquad\forall i\in\M\\
		e^{A_i^\top\tau}\wt P_j e^{A_i\tau}-\wt P_i &\prec0\qquad\forall (i,j)\in\M^2, i\neq j.\label{eq:lmi_jump_gero}
\end{align}
\end{subequations}
Indeed, the fact that \eqref{eq:lmi_flow}-\eqref{eq:lmi_jump} imply \eqref{eq:lmi_flow_gero}-\eqref{eq:lmi_jump_gero} follows from the fact that the latter pair constitutes the matrix blocks (1,1) of the former, taking $\wt P_i=e^{A_i^\top\tau}P_i e^{A_i\tau},~i\in\M,$ and performing simple algebraic manipulations. Additionally, whenever  \eqref{eq:lmi_flow_gero}-\eqref{eq:lmi_jump_gero} hold for given $\wt P_i,~i\in\M,$ the inequalities generated by the matrix blocks (1,1) of \eqref{eq:lmi_flow}-\eqref{eq:lmi_jump} also hold for $ P_i=\epsilon e^{-A_i^\top\tau}\wt P_i e^{-A_i\tau},~i\in\M,$ with $\epsilon>0$. Also, for any $c_i,~i\in\M,$ these conditions can be fully satisfied for $\epsilon$ small enough, some $W_{ij}>0$ close enough to $0$, and $d_{ij}$ such that $W_{ij}d_{ij}$ equal the off-diagonal terms of the left-hand sides. {In the literature concerning dwell-time linear switched systems, several stability conditions equivalent to~\eqref{eq:lmi_flow_gero}-\eqref{eq:lmi_jump_gero} have been proposed, see for example~\cite{Bri15,YuanLv20,AllSha11,Xiang15} and references therein. Our proposed LMI-conditions~\eqref{eq:lmi_flow}-\eqref{eq:lmi_jump}, while tackling a more challenging problem, have the same region of feasibility of~\eqref{eq:lmi_flow_gero}-\eqref{eq:lmi_jump_gero} with respect to the dwell-time parameter $\tau>0$ and the matrices $A_i,~i\in\M$. On the other hand, conditions~\eqref{eq:lmi_flow_gero}-\eqref{eq:lmi_jump_gero} are known to be conservative as they restrict the choice of norms in~Lemma~\ref{lemma:DwellTimeConverseLyapLinear} to \emph{quadratic} norms. To asymptotically reach tight conditions for stability analysis of dwell-time linear switched systems one could consider SOS or polyhedral Lyapunov functions (see~\cite{CheCol12} and~\cite{BlaCol10}, respectively) increasing the computational complexity. Of course these approaches can be adapted in our switched affine systems framework, but for simplicity we do not proceed in this direction.} \hfill $\triangle$
\end{rem}

In contrast with \eqref{eq:lmi_flow_gero}-\eqref{eq:lmi_jump_gero}, the direct verification of \eqref{eq:lmi_flow}-\eqref{eq:lmi_jump} is difficult to be performed numerically due to the product of many decision variables. However, selecting the center of the functions $\wt v_i:\R^n\to \R$ at the equilibrium of the $i$-th subsystem (i.e., choosing $c_i=x_{ei}$) {allows us to restate it as an LMI problem given as:}
\begin{subequations}
  		\begin{align}\label{eq:obj_dwell}
		\min_{P_i,M_{ij},d_{ij}} \sum_{(i,j)\in\M^2} {\rm Tr}~ M_{ij},~~&\text{\rm s.t.}\\
		\begin{bmatrix}
		-Q_{ij}(P_i,P_j)+D &T(d_{ij})\\
		T(d_{ij})^\top & M_{ij}
		\end{bmatrix}&\succ0,\qquad\forall (i,j)\in\M^2\label{eq:lmi_dwell}
		\end{align}
\end{subequations}
where $T(d_{ij})^\top = [I~~d_{ij}]$, $D = {\rm diag}(0,\cdots,0,1)$ and
\[
Q_{ij}(P_i,P_j) = \left\{ \begin{array}{cl}
  \A_i^\top \P_i+ \P_i\A_i,   &~{\rm if}~ i=j \\
  e^{\A_i^\top\tau} \P_i e^{\A_i\tau}- \P_j,   &~{\rm if}~ i\neq j 
\end{array}  \right. .
\]
By applying the Schur Complement Lemma (see~\cite[Section A.5.5]{boyd2004convex}) with respect to $M_{ij}$ in the LMI~\eqref{eq:lmi_dwell}, one can easily verify that these conditions are equivalent to~\eqref{eq:lmi_flow}-\eqref{eq:lmi_jump} by taking $W_{ij}=M_{ij}^{-1}$. The objective function~\eqref{eq:obj_dwell} seeks to indirectly reduce the size of $\V$ by minimizing the sum of the squared lengths of the semi-axes of the ellipsoids $\X_{i},~i\in\M$, which provides a tight estimation of the minimum bounding region. Indeed, exactly minimizing the volume of $\V$ is {overwhelmingly difficult} since it is a non-convex set constituted by the union of uncountable many ellipsoids, recall the definition in~\eqref{eq:def_V}. 

Once a solution $P_i,~M_{ij},~d_{ij},~(i,j)\in\M^2$ has been found, the regions $\X_i,~i\in\M$ can be obtained by solving 
\begin{subequations}
  		\begin{align}\label{eq:obj_Xi}
		\min_{\beta_{ij}\geq 0,\gamma\geq 0} \gamma,~~&\text{\rm s.t.}\\
        \beta_{ij} \mathcal{E}_{ij}-\P_j+\gamma D 		
		&\succ0,\qquad\forall (i,j)\in\M^2,\label{eq:lmi_Xi}
		\end{align}
\end{subequations}
with $\mathcal{E}_{ij}$ defined in \eqref{eq:def_cals} and $W_{ij}=M_{ij}^{-1}$. According to the discussions presented in \cite[Section~3.7.1]{boyd1994linear}, taking $R_\X = \sqrt{\gamma}$ implies that 
$\X_i\supseteq \Eb_{ij}=\{x\in\R^n~\vert~\xi^\top \E_{ij}\xi\leq 0 \},~\forall (i,j)\in\M^2$. Finally, once the region $\X_i$ is characterized, for an arbitrary point $x\in\R^n$ the test whether $x\in\V$ can be performed without major difficulties, as described in the next proposition.

\begin{prop}
A point $x\in\R^n$ is contained in the bounding region $\V$ defined in~\eqref{eq:def_V}  {with norms $v_1,...,v_M$ given in~\eqref{eq:quad_norms_i}}, if and only if there exist $t\in[0,\tau]$ and a pair $(i,j)\in\M^2$ such that
\begin{equation}
\big(x-e^{A_j t}(c_i-x_{ej})-x_{ej}\big)^\top e^{-A_j^\top t}P_ie^{-A_j t} \big(x-e^{A_j t}(c_i-x_{ej})-x_{ej}\big)\leq R_\X^2.\label{eq:x_in_V_cond}
\end{equation}

\end{prop}
\begin{proof}
By definition, $x\in\V$ if for some pair $(i,j)\in\M$ there exists $x_0\in\X_i$ and $t\in[0,\tau]$ such that $x=\Psi_j(t,x_0)$. Notice that, for all $x_0\in\X_i$ we have
\begin{align}
    R_\X^2&\geq  (x_0-c_i)^\top P_i(x_0-c_i)\nonumber\\
    &=(x_0-c_i)^\top e^{A_j^\top t}e^{-A_j^\top t} P_ie^{-A_j t}e^{A_j t}(x_0-c_i)\label{eq:proof_x_in_V}
\end{align}
and also that
\begin{align}
e^{A_j t}(x_0-c_i) &=e^{A_j t}(x_0-x_{ej}-c_i+x_{ej})\nonumber\\
&= \Psi_j(t,x_0)-x_{ej} -e^{A_j t}(c_i-x_{ej})
\end{align}
because one can rewrite $\Psi_j(t,x_0) = e^{A_jt}(x_0-x_{ej}) +x_{ej}$. Replacing this last equality in~\eqref{eq:proof_x_in_V} yields~\eqref{eq:x_in_V_cond} for $x=\Psi_j(t,x_0)$, concluding the proof.
\end{proof}
It is noteworthy to mention that~\eqref{eq:x_in_V_cond} can be easily verified through a line search for $t\in[0,\tau]$ during which the inequality is evaluated for each pair $(i,j)\in\M^2$. The next example illustrates the overall procedure presented in this section. 
\begin{example}
Let us adapt Example~1 in~\cite{geromel2006stability} by introducing affine terms and consider a switched affine system~\eqref{eq:Switchingsystems} defined by
\begin{equation}A_1=\begin{bmatrix}
      	0 & 1 \\ 
      	 -10 & -1
     \end{bmatrix},~ A_2=\begin{bmatrix}
      	0 & 1 \\ 
      	 -0.1 & -0.5
     \end{bmatrix},~ b_1=\begin{bmatrix}
      	-1 \\ 
      	 -1
     \end{bmatrix},~ b_2=\begin{bmatrix}
      	1 \\ 
      	 0
     \end{bmatrix}  
\end{equation}
Solving the optimization problem \eqref{eq:obj_dwell}-\eqref{eq:lmi_dwell} for a minimum dwell-time of $\tau=2.76$ yields the matrices 
\[
P_1=\begin{bmatrix}
      	6.4537 & 0.1019 \\ 
      	 0.1019 & 0.8028
     \end{bmatrix}\times10^{-3},~P_2=\begin{bmatrix}
      	0.4290 & 0.0320 \\ 
      	 0.0320 & 3.0485
     \end{bmatrix}\times10^{-3} 
\]
as solution. This is the least dwell-time value for which we have obtained a feasible solution and it matches the one found in \cite{geromel2006stability} for the linearized system~\eqref{eq:DeAffinedSystem}, as proved in Remark~\ref{rem:WeAREstrongas,JOse}. To obtain the regions $\X_i,~i\in\M$ and subsequently the bounding region $\V$, we solved the optimization problem~\eqref{eq:obj_Xi}-\eqref{eq:lmi_Xi} and obtained $R_\X = \sqrt{\gamma}=22.3274$ which allows us to draw $\X_1,~\X_2,~\V_1$ and $\V_2$ as given in Figure~\ref{fig:dwell_bounding_geromel}. These same regions are also represented in Figure~\ref{fig:dwell_bounding_geromel} for $\tau=5$ and $\tau=10$. The bounding region $\V=\V_1\cup\V_2$ ensures by Theorem~\ref{thm:DwellTimeBoundedness} that, for any initial condition $x\in\R^n$, we have the solution $\Psi_\sigma(t,x)\in \V$ for all $t$ after a finite amount of time and all $\sigma\in \cS_{\dw}(\tau)$. This figure also illustrates that the area of these regions decreases as $\tau$ increases, as discussed in Remark~\ref{rem:rho_dt_sets}.
\begin{figure}
    \centering
	\def\svgwidth{.32\linewidth}
    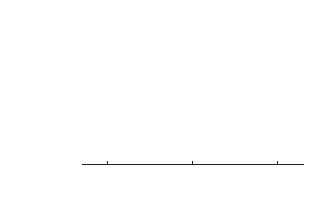
	\def\svgwidth{.32\linewidth}
    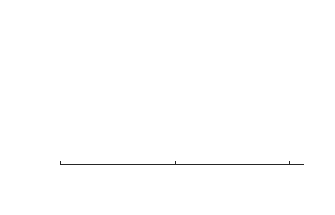
	\def\svgwidth{.32\linewidth}
    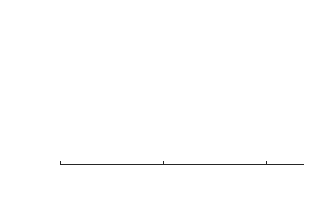
    \caption{For three different values of dwell-time $\tau$, the regions $X_1$ (red line), $X_2$ (blue line), $\V_1$ (red area) and $\V_2$ (blue area) are represented. The bounding region $\V=\V_1\cup\V_2$ is to where all solutions converge under any dwell-time switching signal $\sigma\in \cS_{\dw}(\tau)$. The black points represent the equilibria of each subsystem and the dotted regions represent $\Phi_2(\tau,\X_1)$ and $\Phi_1(\tau,\X_2)$ keeping the same color pattern.}
    \label{fig:dwell_bounding_geromel}
\end{figure}
\end{example}
The above example showed how the methodology developed in this paper allows the estimation of bounding regions for switched affine systems under dwell-time switching. In the following section, we present a path-following algorithm based on~\cite{hassibi1999path} that is shown to provide more accurate bounding regions. 

\subsubsection{Path-following method to enhance the ellipsoid positioning}
Inspired by the methodology presented in~\cite{hassibi1999path}, we derive a local optimization procedure for selecting the centers $c_i$ of the functions $\wt v_i(x) =\sqrt{(x-c_i)^\top P_i(x-c_i)}, ~i\in\M$. Particularly, the present problem is well adapted to be solved by this path-following method given that the choice of $c_i\in\R^n,~i\in\M$ does not interfere in the feasibility of the optimization problem in~\eqref{eq:obj_dwell}-\eqref{eq:lmi_dwell}. Also, though this method only guarantees convergence to local optima, the fact that suitable points $c_i\in\R^n,~i\in\M$, should lie close to $x_{ei}$ (tending to it when $\tau\rightarrow \infty$) allows us to efficiently warm-start this method. {As described in~\cite{hassibi1999path}, the idea behind the path-following algorithm is to linearize the non-linear contraints using a first-order approximation around a given feasible solution and iteratively compute a direction in the decision space that slightly improves the objective function.} Before presenting the algorithm adapted to our context, let us use~\eqref{eq:obj_dwell} and~\eqref{eq:lmi_dwell} to define the following objective function
\begin{equation}
    f(M_{ij})=\sum_{(i,j)\in\M^2} {\rm Tr}~ M_{ij}
\end{equation}
and the matrix-valued function
\begin{equation}
    h_{ij}(P_i,P_j,c_i,c_j,M_{ij},d_{ij}) = \begin{bmatrix}
		-Q_{ij}(P_i,P_j)+D &T(d_{ij})\\
		T(d_{ij})^\top & M_{ij}
		\end{bmatrix}
\end{equation}
where the dependence on $c_i$ and $c_j$ happens through the the matrices $\P_i$ and $\P_j$ in the definition of $Q_{ij}$. 

\begin{algorithm}[tb]
	  \caption{
		    Path-following method for optimizing the centers $c_i,~i\in\M$
		  }
	  \label{algo:path_following}
\begin{algorithmic}[1]
	\REQUIRE{System matrices $(A_i,b_i),~i\in\M$, a dwell-time $\tau$, a precision $\varepsilon>0$ and a step-bound $\delta>0$}
	\STATE{$c_i \gets x_{ei}, ~\forall i \in \M$} \label{line:init_1}
	\STATE{$(P_i,M_{ij},d_{ij})\gets \arg\min_{P_i,M_{ij},d_{ij}} f(M_{ij}) $~ s.t.~\eqref{eq:lmi_dwell}}\label{line:init_2}
	\REPEAT
	   	\STATE{$(\hat P_i,\hat c_i, \hat M_{ij},\hat d_{ij})\gets \arg\min_{ \hat P_i, \hat c_i,  \hat M_{ij}, \hat d_{ij}} f( \hat M_{ij}) $~ s.t.\label{line:lin_opt_prob}\\
		   {$$h_{ij}(P_i,P_j,c_i,c_j,M_{ij},d_{ij}) +\sum_{\substack{X\in\{ P_i, P_j, c_i,\\\;\;\; c_j, M_{ij}, d_{ij}\}}}\frac{\partial h_{ij}(P_i,P_j,c_i,c_j,M_{ij},d_{ij})}{\partial X}(\hat{X}-X)>0,~ $$
		   $$  \hspace{5.2cm}\text{for } (i,j)\in\M^2,$$
		   $$-\delta<\hat{X}-X<\delta,\quad\hspace{2cm}\text{for } X\in( P_i, c_i),$$}}%
		   
		\STATE{ $c_i\gets c_i-\hat c_i$}
 		\STATE{$(P_i,M_{ij},d_{ij})\gets \arg\min_{P_i,M_{ij},d_{ij}} f(M_{ij}) $~ s.t.~\eqref{eq:lmi_dwell}}
	\UNTIL{$|f(M_{ij})-f(\hat M_{ij})|>\epsilon f(M_{ij})$}

	\RETURN $( P_i, c_i,  M_{ij}, d_{ij})$
\end{algorithmic}
\end{algorithm}
%
%
%
%
%
Algorithm~\ref{algo:path_following} performs the optimization with respect to the centers $c_i,~i\in\M$. With some abuse of notation, we refer to the solution tuple $(P_i,c_i,M_{ij},d_{ij})_{(i,j)\in\M^2}$ simply as $(P_i,c_i,M_{ij},d_{ij})$. In Lines~\ref{line:init_1} and~\ref{line:init_2} the algorithm initializes each $c_i$ with the associated equilibrium $x_{ei}$ and the variables $P_i,M_{ij},d_{ij}$ with the corresponding solution to \eqref{eq:obj_dwell}-\eqref{eq:lmi_dwell}. Then, in Line~\ref{line:lin_opt_prob}, we solve the first-order approximation of \eqref{eq:obj_dwell}-\eqref{eq:lmi_dwell} taking into account $c_i$ as variables. The first constraint characterizes the linear approximation of the nonlinear constraint \eqref{eq:lmi_dwell} around the point defined by the current variables $(P_i,c_i,M_{ij},d_{ij})$ and the second constraint should be interpreted as an element-wise bound on the difference between $(P_i,c_i)$ and the ``perturbed'' variables $(\hat P_i,\hat c_i)$. This constraint is not applied to $\hat M_{ij}$ and $\hat d_{ij}$ as the original problem is convex with respect to them. The control parameter $\delta=10^{-1}$ was chosen in our numerical experiments, which was sufficiently small such that the first-order approximation is valid and large enough such that convergence is met in a few iterations. Afterward, another feasible point $(P_i,c_i,M_{ij},d_{ij})$ is determined by solving~\eqref{eq:obj_dwell}-\eqref{eq:lmi_dwell} for new centers $c_i$ and this procedure is repeated until a convergence criterion is satisfied. The next example illustrates cases where the use of Algorithm~\ref{algo:path_following} allows the determination of tighter bounding regions.

\begin{example}
Consider a switched affine system~\eqref{eq:Switchingsystems} defined by
{\small\[
A_1=\begin{bmatrix}
      	-5 & 1 \\ 
      	 -1 & -4
     \end{bmatrix},~A_2= 
      \begin{bmatrix}
      	-5 & -1 \\ 
      	 1 & -4
     \end{bmatrix},~A_3= 
      \begin{bmatrix}
      	-2 & 8 \\ 
      	 -5 & -5
     \end{bmatrix},~b_1=\begin{bmatrix}
      	-50 \\ 
      	 -10
     \end{bmatrix},~b_2=
      \begin{bmatrix}
      	-10 \\ 
      	 -40
     \end{bmatrix},~b_3=0.
\]}
The equilibrium points of each subsystem are $x_{e1}=[-10~~0]^\top$, $x_{e2}=[0~~-10]^\top$ and $x_{e3}=0$. For two different values of dwell-time $\tau\in\{0.1, 0.5\}$, we applied Algorithm~\ref{algo:path_following} to determine suitable centers $c_i,~i\in\M,$ for the functions $\wt v_i(x)$ used to define the sets $\X_i,~i\in\M,$ and the bounding region $\V$, given in~\eqref{eq:def_V}. Considering $\epsilon=10^{-3}$ For $\tau=0.1$, after $48$ iterations the algorithm converged to the centers
{\small\[
c_1=\begin{bmatrix}
      	-6.9714 \\ 
      	 -1.9481
     \end{bmatrix},~c_2= 
      \begin{bmatrix}
      	-0.2497 \\ 
      	 -7.2263
     \end{bmatrix},~c_3=
      \begin{bmatrix}
      	-0.9154 \\ 
      	 -1.5790
     \end{bmatrix} 
\]}
and for $\tau=0.5$, after 8 iterations the algorithm converged to
{\small\[
c_1=\begin{bmatrix}
      	-9.7879 \\ 
      	 -0.5505
     \end{bmatrix},~c_2= 
      \begin{bmatrix}
      	-0.2918 \\ 
      	 -9.5003
     \end{bmatrix},~c_3= 
      \begin{bmatrix}
      	0.3777 \\ 
      	 0.3399
     \end{bmatrix} 
\]}
The value of the objective function $f(W_{ij})$ over each iteration for each case is depicted in Figure~\ref{fig:obj_func_x_iter}. Notice that for the smaller dwell-time value, optimizing the centers allowed us to reduce more significantly the value of the objective function, which indicates that, for the larger dwell-time, the equilibria $x_{ei}$ were closer to the locally best points $c_i,~i\in\M$.  The corresponding regions $\X_i$ and $\V_i$ are given in Figure~\ref{fig:reg_X_V_dwell_2}.

\begin{figure}
    \centering
	\def\svgwidth{.46\linewidth}
    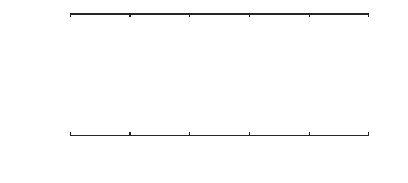
	\def\svgwidth{.51\linewidth}
    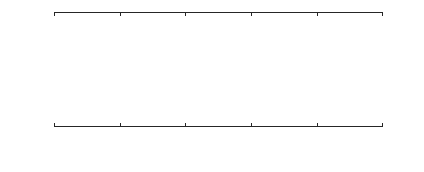
    \caption{In the course of several iterations of Algorithm~\ref{algo:path_following}, the value of the objective function $f(M_{ij})$ is depicted for $\tau=0.1$ (left) and $\tau=0.5$ (right).}
    \label{fig:obj_func_x_iter}
\end{figure}
\begin{figure}
    \centering
	\def\svgwidth{.4\linewidth}
    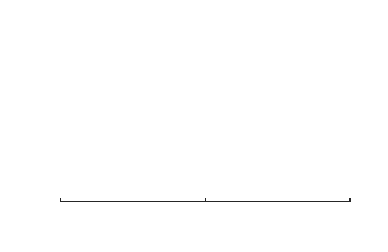
	\def\svgwidth{.4\linewidth}
    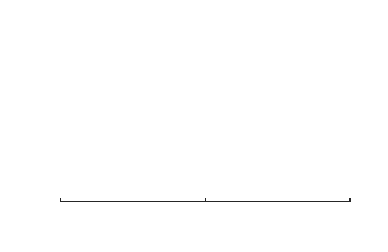 
	\def\svgwidth{.4\linewidth}
    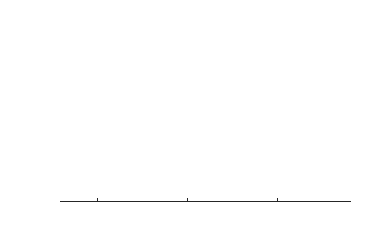
	\def\svgwidth{.4\linewidth}
    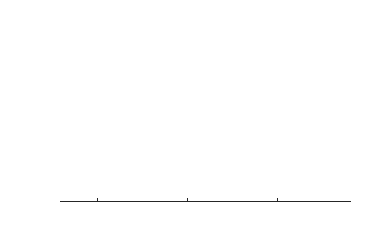 
    \caption{Regions $\X_1,~\X_2$ and $\X_3$ (red, blue and green lines, resp.) and regions $\V_1,~\V_2$ and $\V_3$ (red, blue and green areas, resp.) for $\tau=0.1$ (top) and $\tau=0.5$ (bottom). The figures on the left take $c_i=x_{ei}$ for all $i\in\M$ whereas the ones on the right have $c_i$ obtained from Algorithm~\ref{algo:path_following}.  }
    \label{fig:reg_X_V_dwell_2}
\end{figure}
\end{example}
\section{Conclusion}\label{sec:Concl}

Stability properties of continuous-time switched affine systems under arbitrary and dwell-time switching were discussed in this paper. We demonstrated that, in both cases, the boundedness of the state trajectories is closely related to the global asymptotic stability of the associated linearized system, i.e., when the affine terms are ignored. However, differently from switched linear systems, switched affine systems are not stable with respect to a point but rather to sets. In the arbitrary switching case, we characterized the existence of forward invariant sets. On the other hand, when dwell-time switching is considered, forward invariant sets need not exist and bounding regions are considered. Theoretical results ensuring the existence and non-existence of such sets were given and numerical methods based on convex optimization were devised to outer approximate them. These results were illustrated by numerical examples for each case.

In the future, we plan to study the stability of switched affine systems under other classes of switching signals such as periodic, path-constrained, and Markov-jump switching. {The case of state-dependent switching systems will also be considered, since it naturally arises considering affine dynamics as representation of a local first-order approximation of general smooth vector-fields. Moreover, this shall provide novel insights into the analysis of symbolic-abstraction based systems, for which local approximations techniques represent a central tool.}
\bibliography{biblio.bib} 
\bibliographystyle{ieeetr}
\end{document}

%% file: counter_ex.pdf_tex
\begingroup%
  \makeatletter%
  \providecommand\color[2][]{%
    \errmessage{(Inkscape) Color is used for the text in Inkscape, but the package 'color.sty' is not loaded}%
    \renewcommand\color[2][]{}%
  }%
  \providecommand\transparent[1]{%
    \errmessage{(Inkscape) Transparency is used (non-zero) for the text in Inkscape, but the package 'transparent.sty' is not loaded}%
    \renewcommand\transparent[1]{}%
  }%
  \providecommand\rotatebox[2]{#2}%
  \newcommand*\fsize{\dimexpr\f@size pt\relax}%
  \newcommand*\lineheight[1]{\fontsize{\fsize}{#1\fsize}\selectfont}%
  \ifx\svgwidth\undefined%
    \setlength{\unitlength}{287.9999928bp}%
    \ifx\svgscale\undefined%
      \relax%
    \else%
      \setlength{\unitlength}{\unitlength * \real{\svgscale}}%
    \fi%
  \else%
    \setlength{\unitlength}{\svgwidth}%
  \fi%
  \global\let\svgwidth\undefined%
  \global\let\svgscale\undefined%
  \makeatother \scriptsize%
  \begin{picture}(1,0.8828125)%
    \lineheight{1}%
    \setlength\tabcolsep{0pt}%
    \put(0,0){\includegraphics[width=\unitlength,page=1]{counter_ex.pdf}}%
    \put(0.24109102,0.0616319){\color[rgb]{0.14901961,0.14901961,0.14901961}\makebox(0,0)[lt]{\lineheight{1.25}\smash{\begin{tabular}[t]{l}-1\end{tabular}}}}%
    \put(0.4276316,0.0616319){\color[rgb]{0.14901961,0.14901961,0.14901961}\makebox(0,0)[lt]{\lineheight{1.25}\smash{\begin{tabular}[t]{l}-0.5\end{tabular}}}}%
    \put(0.65193255,0.0616319){\color[rgb]{0.14901961,0.14901961,0.14901961}\makebox(0,0)[lt]{\lineheight{1.25}\smash{\begin{tabular}[t]{l}0\end{tabular}}}}%
    \put(0.83717109,0.0616319){\color[rgb]{0.14901961,0.14901961,0.14901961}\makebox(0,0)[lt]{\lineheight{1.25}\smash{\begin{tabular}[t]{l}0.5\end{tabular}}}}%
    \put(0.49869826,0.01041667){\color[rgb]{0.14901961,0.14901961,0.14901961}\makebox(0,0)[lt]{\lineheight{1.25}\smash{\begin{tabular}[t]{l}\normalsize $x_1$\end{tabular}}}}%
    \put(0,0){\includegraphics[width=\unitlength,page=2]{counter_ex.pdf}}%
    \put(0.08767365,0.20949835){\color[rgb]{0.14901961,0.14901961,0.14901961}\makebox(0,0)[lt]{\lineheight{1.25}\smash{\begin{tabular}[t]{l}-1\end{tabular}}}}%
    \put(0.05642365,0.39590187){\color[rgb]{0.14901961,0.14901961,0.14901961}\makebox(0,0)[lt]{\lineheight{1.25}\smash{\begin{tabular}[t]{l}-0.5\end{tabular}}}}%
    \put(0.10069449,0.58230538){\color[rgb]{0.14901961,0.14901961,0.14901961}\makebox(0,0)[lt]{\lineheight{1.25}\smash{\begin{tabular}[t]{l}0\end{tabular}}}}%
    \put(0.06684032,0.7687089){\color[rgb]{0.14901961,0.14901961,0.14901961}\makebox(0,0)[lt]{\lineheight{1.25}\smash{\begin{tabular}[t]{l}0.5\end{tabular}}}}%
    \put(0.03645833,0.44010451){\color[rgb]{0.14901961,0.14901961,0.14901961}\rotatebox{90.00000248}{\makebox(0,0)[lt]{\lineheight{1.25}\smash{\begin{tabular}[t]{l}\normalsize $x_2$\end{tabular}}}}}%
    \put(0,0){\includegraphics[width=\unitlength,page=3]{counter_ex.pdf}}%
    \put(0.7265625,0.73697917){\color[rgb]{0,0,0}\makebox(0,0)[lt]{\lineheight{1.25}\smash{\begin{tabular}[t]{l}\normalsize \scriptsize{}Fil$_0(\mathcal{F})$\end{tabular}}}}%
    \put(0,0){\includegraphics[width=\unitlength,page=4]{counter_ex.pdf}}%
    \put(0.7265625,0.67578125){\color[rgb]{0,0,0}\makebox(0,0)[lt]{\lineheight{1.25}\smash{\begin{tabular}[t]{l}\normalsize \scriptsize$\mathcal{Z}$\end{tabular}}}}%
    \put(0,0){\includegraphics[width=\unitlength,page=5]{counter_ex.pdf}}%
    \put(0.7265625,0.61458333){\color[rgb]{0,0,0}\makebox(0,0)[lt]{\lineheight{1.25}\smash{\begin{tabular}[t]{l}\normalsize \scriptsize$A_1x+b_1$\end{tabular}}}}%
    \put(0,0){\includegraphics[width=\unitlength,page=6]{counter_ex.pdf}}%
    \put(0.7265625,0.55338542){\color[rgb]{0,0,0}\makebox(0,0)[lt]{\lineheight{1.25}\smash{\begin{tabular}[t]{l}\normalsize \scriptsize$A_2x+b_2$\end{tabular}}}}%
    \put(0,0){\includegraphics[width=\unitlength,page=7]{counter_ex.pdf}}%
    \put(0.7265625,0.4921875){\color[rgb]{0,0,0}\makebox(0,0)[lt]{\lineheight{1.25}\smash{\begin{tabular}[t]{l}\normalsize \scriptsize$x^*$\end{tabular}}}}%
    \put(0,0){\includegraphics[width=\unitlength,page=8]{counter_ex.pdf}}%
  \end{picture}%
\endgroup%

%% file: arb_swi_Kinf.pdf_tex
\begingroup%
  \makeatletter%
  \providecommand\color[2][]{%
    \errmessage{(Inkscape) Color is used for the text in Inkscape, but the package 'color.sty' is not loaded}%
    \renewcommand\color[2][]{}%
  }%
  \providecommand\transparent[1]{%
    \errmessage{(Inkscape) Transparency is used (non-zero) for the text in Inkscape, but the package 'transparent.sty' is not loaded}%
    \renewcommand\transparent[1]{}%
  }%
  \providecommand\rotatebox[2]{#2}%
  \newcommand*\fsize{\dimexpr\f@size pt\relax}%
  \newcommand*\lineheight[1]{\fontsize{\fsize}{#1\fsize}\selectfont}%
  \ifx\svgwidth\undefined%
    \setlength{\unitlength}{599.999985bp}%
    \ifx\svgscale\undefined%
      \relax%
    \else%
      \setlength{\unitlength}{\unitlength * \real{\svgscale}}%
    \fi%
  \else%
    \setlength{\unitlength}{\svgwidth}%
  \fi%
  \global\let\svgwidth\undefined%
  \global\let\svgscale\undefined%
  \makeatother \scriptsize%
  \begin{picture}(1,0.37375)%
    \lineheight{1}%
    \setlength\tabcolsep{0pt}%
    \put(0,0){\includegraphics[width=\unitlength,page=1]{arb_swi_Kinf.pdf}}%
    \put(0.12280405,0.02958331){\color[rgb]{0.14901961,0.14901961,0.14901961}\makebox(0,0)[lt]{\lineheight{1.25}\smash{\begin{tabular}[t]{l}-1.5\end{tabular}}}}%
    \put(0.22084458,0.02958331){\color[rgb]{0.14901961,0.14901961,0.14901961}\makebox(0,0)[lt]{\lineheight{1.25}\smash{\begin{tabular}[t]{l}-1\end{tabular}}}}%
    \put(0.30388512,0.02958331){\color[rgb]{0.14901961,0.14901961,0.14901961}\makebox(0,0)[lt]{\lineheight{1.25}\smash{\begin{tabular}[t]{l}-0.5\end{tabular}}}}%
    \put(0.40505067,0.02958331){\color[rgb]{0.14901961,0.14901961,0.14901961}\makebox(0,0)[lt]{\lineheight{1.25}\smash{\begin{tabular}[t]{l}0\end{tabular}}}}%
    \put(0.28625015,0.00541669){\color[rgb]{0.14901961,0.14901961,0.14901961}\makebox(0,0)[lt]{\lineheight{1.25}\smash{\begin{tabular}[t]{l}\normalsize $x_1$\end{tabular}}}}%
    \put(0,0){\includegraphics[width=\unitlength,page=2]{arb_swi_Kinf.pdf}}%
    \put(0.09083335,0.05481419){\color[rgb]{0.14901961,0.14901961,0.14901961}\makebox(0,0)[lt]{\lineheight{1.25}\smash{\begin{tabular}[t]{l}-1.5\end{tabular}}}}%
    \put(0.10583335,0.13420608){\color[rgb]{0.14901961,0.14901961,0.14901961}\makebox(0,0)[lt]{\lineheight{1.25}\smash{\begin{tabular}[t]{l}-1\end{tabular}}}}%
    \put(0.09083335,0.21359798){\color[rgb]{0.14901961,0.14901961,0.14901961}\makebox(0,0)[lt]{\lineheight{1.25}\smash{\begin{tabular}[t]{l}-0.5\end{tabular}}}}%
    \put(0.11208335,0.29298985){\color[rgb]{0.14901961,0.14901961,0.14901961}\makebox(0,0)[lt]{\lineheight{1.25}\smash{\begin{tabular}[t]{l}0\end{tabular}}}}%
    \put(0.08083333,0.18812512){\color[rgb]{0.14901961,0.14901961,0.14901961}\rotatebox{90.00000248}{\makebox(0,0)[lt]{\lineheight{1.25}\smash{\begin{tabular}[t]{l}\normalsize $x_2$\end{tabular}}}}}%
    \put(0,0){\includegraphics[width=\unitlength,page=3]{arb_swi_Kinf.pdf}}%
    \put(0.4125,0.3285779){\color[rgb]{0,0,0}\makebox(0,0)[lt]{\lineheight{1.25}\smash{\begin{tabular}[t]{l}\normalsize \scriptsize$K_{\rm Q}$\end{tabular}}}}%
    \put(0,0){\includegraphics[width=\unitlength,page=4]{arb_swi_Kinf.pdf}}%
    \put(0.4125,0.3089221){\color[rgb]{0,0,0}\makebox(0,0)[lt]{\lineheight{1.25}\smash{\begin{tabular}[t]{l}\normalsize \scriptsize$K_{\rm SOS}$\end{tabular}}}}%
    \put(0,0){\includegraphics[width=\unitlength,page=5]{arb_swi_Kinf.pdf}}%
    \put(0.564375,0.02958331){\color[rgb]{0.14901961,0.14901961,0.14901961}\makebox(0,0)[lt]{\lineheight{1.25}\smash{\begin{tabular}[t]{l}0\end{tabular}}}}%
    \put(0.648125,0.02958331){\color[rgb]{0.14901961,0.14901961,0.14901961}\makebox(0,0)[lt]{\lineheight{1.25}\smash{\begin{tabular}[t]{l}2\end{tabular}}}}%
    \put(0.731875,0.02958331){\color[rgb]{0.14901961,0.14901961,0.14901961}\makebox(0,0)[lt]{\lineheight{1.25}\smash{\begin{tabular}[t]{l}4\end{tabular}}}}%
    \put(0.815625,0.02958331){\color[rgb]{0.14901961,0.14901961,0.14901961}\makebox(0,0)[lt]{\lineheight{1.25}\smash{\begin{tabular}[t]{l}6\end{tabular}}}}%
    \put(0.899375,0.02958331){\color[rgb]{0.14901961,0.14901961,0.14901961}\makebox(0,0)[lt]{\lineheight{1.25}\smash{\begin{tabular}[t]{l}8\end{tabular}}}}%
    \put(0.71625015,0.00541669){\color[rgb]{0.14901961,0.14901961,0.14901961}\makebox(0,0)[lt]{\lineheight{1.25}\smash{\begin{tabular}[t]{l}\normalsize \!\!\!time $(t)$\end{tabular}}}}%
    \put(0,0){\includegraphics[width=\unitlength,page=6]{arb_swi_Kinf.pdf}}%
    \put(0.52875,0.05625){\color[rgb]{0.14901961,0.14901961,0.14901961}\makebox(0,0)[lt]{\lineheight{1.25}\smash{\begin{tabular}[t]{l}10\end{tabular}}}}%
    \put(0.55,0.06375){\color[rgb]{0.14901961,0.14901961,0.14901961}\makebox(0,0)[lt]{\lineheight{1.25}\smash{\begin{tabular}[t]{l}-2\end{tabular}}}}%
    \put(0.53375,0.19375){\color[rgb]{0.14901961,0.14901961,0.14901961}\makebox(0,0)[lt]{\lineheight{1.25}\smash{\begin{tabular}[t]{l}10\end{tabular}}}}%
    \put(0.555,0.20125){\color[rgb]{0.14901961,0.14901961,0.14901961}\makebox(0,0)[lt]{\lineheight{1.25}\smash{\begin{tabular}[t]{l}0\end{tabular}}}}%
    \put(0.53375,0.3325){\color[rgb]{0.14901961,0.14901961,0.14901961}\makebox(0,0)[lt]{\lineheight{1.25}\smash{\begin{tabular}[t]{l}10\end{tabular}}}}%
    \put(0.555,0.34){\color[rgb]{0.14901961,0.14901961,0.14901961}\makebox(0,0)[lt]{\lineheight{1.25}\smash{\begin{tabular}[t]{l}2\end{tabular}}}}%
    \put(0.51833331,0.18187512){\color[rgb]{0.14901961,0.14901961,0.14901961}\rotatebox{90.00000248}{\makebox(0,0)[lt]{\lineheight{1.25}\smash{\begin{tabular}[t]{l}\normalsize \!\!$V\big(x(t)\big)$\end{tabular}}}}}%
    \put(0,0){\includegraphics[width=\unitlength,page=7]{arb_swi_Kinf.pdf}}%
    \put(0.74604165,0.18733333){\color[rgb]{0.14901961,0.14901961,0.14901961}\makebox(0,0)[lt]{\lineheight{1.25}\smash{\begin{tabular}[t]{l}2\end{tabular}}}}%
    \put(0.7903125,0.18733333){\color[rgb]{0.14901961,0.14901961,0.14901961}\makebox(0,0)[lt]{\lineheight{1.25}\smash{\begin{tabular}[t]{l}3\end{tabular}}}}%
    \put(0.83458335,0.18733333){\color[rgb]{0.14901961,0.14901961,0.14901961}\makebox(0,0)[lt]{\lineheight{1.25}\smash{\begin{tabular}[t]{l}4\end{tabular}}}}%
    \put(0.87885415,0.18733333){\color[rgb]{0.14901961,0.14901961,0.14901961}\makebox(0,0)[lt]{\lineheight{1.25}\smash{\begin{tabular}[t]{l}5\end{tabular}}}}%
    \put(0.923125,0.18733333){\color[rgb]{0.14901961,0.14901961,0.14901961}\makebox(0,0)[lt]{\lineheight{1.25}\smash{\begin{tabular}[t]{l}6\end{tabular}}}}%
    \put(0,0){\includegraphics[width=\unitlength,page=8]{arb_swi_Kinf.pdf}}%
    \put(0.70358329,0.23610027){\color[rgb]{0.14901961,0.14901961,0.14901961}\makebox(0,0)[lt]{\lineheight{1.25}\smash{\begin{tabular}[t]{l}8\end{tabular}}}}%
    \put(0.69483329,0.27691108){\color[rgb]{0.14901961,0.14901961,0.14901961}\makebox(0,0)[lt]{\lineheight{1.25}\smash{\begin{tabular}[t]{l}10\end{tabular}}}}%
    \put(0.69483329,0.31025596){\color[rgb]{0.14901961,0.14901961,0.14901961}\makebox(0,0)[lt]{\lineheight{1.25}\smash{\begin{tabular}[t]{l}12\end{tabular}}}}%
    \put(0,0){\includegraphics[width=\unitlength,page=9]{arb_swi_Kinf.pdf}}%
    \put(0.72625,0.32875){\color[rgb]{0.14901961,0.14901961,0.14901961}\makebox(0,0)[lt]{\lineheight{1.25}\smash{\begin{tabular}[t]{l}10\end{tabular}}}}%
    \put(0.74375,0.335){\color[rgb]{0.14901961,0.14901961,0.14901961}\makebox(0,0)[lt]{\lineheight{1.25}\smash{\begin{tabular}[t]{l}-3\end{tabular}}}}%
    \put(0,0){\includegraphics[width=\unitlength,page=10]{arb_swi_Kinf.pdf}}%
  \end{picture}%
\endgroup%

%% file: gero_dt_min.pdf_tex
\begingroup%
  \makeatletter%
  \providecommand\color[2][]{%
    \errmessage{(Inkscape) Color is used for the text in Inkscape, but the package 'color.sty' is not loaded}%
    \renewcommand\color[2][]{}%
  }%
  \providecommand\transparent[1]{%
    \errmessage{(Inkscape) Transparency is used (non-zero) for the text in Inkscape, but the package 'transparent.sty' is not loaded}%
    \renewcommand\transparent[1]{}%
  }%
  \providecommand\rotatebox[2]{#2}%
  \newcommand*\fsize{\dimexpr\f@size pt\relax}%
  \newcommand*\lineheight[1]{\fontsize{\fsize}{#1\fsize}\selectfont}%
  \ifx\svgwidth\undefined%
    \setlength{\unitlength}{160.49999599bp}%
    \ifx\svgscale\undefined%
      \relax%
    \else%
      \setlength{\unitlength}{\unitlength * \real{\svgscale}}%
    \fi%
  \else%
    \setlength{\unitlength}{\svgwidth}%
  \fi%
  \global\let\svgwidth\undefined%
  \global\let\svgscale\undefined%
  \makeatother \scriptsize%
  \begin{picture}(1,0.64018692)%
    \lineheight{1}%
    \setlength\tabcolsep{0pt}%
    \put(0,0){\includegraphics[width=\unitlength,page=1]{gero_dt_min.pdf}}%
    \put(0.24415265,0.08442368){\color[rgb]{0.14901961,0.14901961,0.14901961}\makebox(0,0)[lt]{\lineheight{1.25}\smash{\begin{tabular}[t]{l}-1000\end{tabular}}}}%
    \put(0.55718731,0.08442368){\color[rgb]{0.14901961,0.14901961,0.14901961}\makebox(0,0)[lt]{\lineheight{1.25}\smash{\begin{tabular}[t]{l}0\end{tabular}}}}%
    \put(0.76274533,0.08442368){\color[rgb]{0.14901961,0.14901961,0.14901961}\makebox(0,0)[lt]{\lineheight{1.25}\smash{\begin{tabular}[t]{l}1000\end{tabular}}}}%
    \put(0.53971994,0.01277266){\color[rgb]{0.14901961,0.14901961,0.14901961}\makebox(0,0)[lt]{\lineheight{1.25}\smash{\begin{tabular}[t]{l}\normalsize $x_1$\end{tabular}}}}%
    \put(0,0){\includegraphics[width=\unitlength,page=2]{gero_dt_min.pdf}}%
    \put(0.07975078,0.13084112){\color[rgb]{0.14901961,0.14901961,0.14901961}\makebox(0,0)[lt]{\lineheight{1.25}\smash{\begin{tabular}[t]{l}-4000\end{tabular}}}}%
    \put(0.07975078,0.23247664){\color[rgb]{0.14901961,0.14901961,0.14901961}\makebox(0,0)[lt]{\lineheight{1.25}\smash{\begin{tabular}[t]{l}-2000\end{tabular}}}}%
    \put(0.19657321,0.33411215){\color[rgb]{0.14901961,0.14901961,0.14901961}\makebox(0,0)[lt]{\lineheight{1.25}\smash{\begin{tabular}[t]{l}0\end{tabular}}}}%
    \put(0.09844237,0.43574766){\color[rgb]{0.14901961,0.14901961,0.14901961}\makebox(0,0)[lt]{\lineheight{1.25}\smash{\begin{tabular}[t]{l}2000\end{tabular}}}}%
    \put(0.09844237,0.53738318){\color[rgb]{0.14901961,0.14901961,0.14901961}\makebox(0,0)[lt]{\lineheight{1.25}\smash{\begin{tabular}[t]{l}4000\end{tabular}}}}%
    \put(0.05638629,0.31542079){\color[rgb]{0.14901961,0.14901961,0.14901961}\rotatebox{90.00000248}{\makebox(0,0)[lt]{\lineheight{1.25}\smash{\begin{tabular}[t]{l}\normalsize $x_2$\end{tabular}}}}}%
    \put(0.55373886,0.58101636){\color[rgb]{0,0,0}\makebox(0,0)[lt]{\lineheight{1.25}\smash{\begin{tabular}[t]{l}\normalsize \!\!\!\!$\tau=2.76$\end{tabular}}}}%
    \put(0,0){\includegraphics[width=\unitlength,page=3]{gero_dt_min.pdf}}%
  \end{picture}%
\endgroup%

%% file: gero_dt_5.pdf_tex
\begingroup%
  \makeatletter%
  \providecommand\color[2][]{%
    \errmessage{(Inkscape) Color is used for the text in Inkscape, but the package 'color.sty' is not loaded}%
    \renewcommand\color[2][]{}%
  }%
  \providecommand\transparent[1]{%
    \errmessage{(Inkscape) Transparency is used (non-zero) for the text in Inkscape, but the package 'transparent.sty' is not loaded}%
    \renewcommand\transparent[1]{}%
  }%
  \providecommand\rotatebox[2]{#2}%
  \newcommand*\fsize{\dimexpr\f@size pt\relax}%
  \newcommand*\lineheight[1]{\fontsize{\fsize}{#1\fsize}\selectfont}%
  \ifx\svgwidth\undefined%
    \setlength{\unitlength}{160.49999599bp}%
    \ifx\svgscale\undefined%
      \relax%
    \else%
      \setlength{\unitlength}{\unitlength * \real{\svgscale}}%
    \fi%
  \else%
    \setlength{\unitlength}{\svgwidth}%
  \fi%
  \global\let\svgwidth\undefined%
  \global\let\svgscale\undefined%
  \makeatother \scriptsize%
  \begin{picture}(1,0.64018692)%
    \lineheight{1}%
    \setlength\tabcolsep{0pt}%
    \put(0,0){\includegraphics[width=\unitlength,page=1]{gero_dt_5.pdf}}%
    \put(0.13785047,0.08442368){\color[rgb]{0.14901961,0.14901961,0.14901961}\makebox(0,0)[lt]{\lineheight{1.25}\smash{\begin{tabular}[t]{l}-10\end{tabular}}}}%
    \put(0.50594977,0.08442368){\color[rgb]{0.14901961,0.14901961,0.14901961}\makebox(0,0)[lt]{\lineheight{1.25}\smash{\begin{tabular}[t]{l}0\end{tabular}}}}%
    \put(0.83199307,0.08442368){\color[rgb]{0.14901961,0.14901961,0.14901961}\makebox(0,0)[lt]{\lineheight{1.25}\smash{\begin{tabular}[t]{l}10\end{tabular}}}}%
    \put(0.50700974,0.01277266){\color[rgb]{0.14901961,0.14901961,0.14901961}\makebox(0,0)[lt]{\lineheight{1.25}\smash{\begin{tabular}[t]{l}\normalsize $x_1$\end{tabular}}}}%
    \put(0,0){\includegraphics[width=\unitlength,page=2]{gero_dt_5.pdf}}%
    \put(0.07975078,0.1968183){\color[rgb]{0.14901961,0.14901961,0.14901961}\makebox(0,0)[lt]{\lineheight{1.25}\smash{\begin{tabular}[t]{l}-20\end{tabular}}}}%
    \put(0.13115265,0.36710078){\color[rgb]{0.14901961,0.14901961,0.14901961}\makebox(0,0)[lt]{\lineheight{1.25}\smash{\begin{tabular}[t]{l}0\end{tabular}}}}%
    \put(0.09844237,0.53738318){\color[rgb]{0.14901961,0.14901961,0.14901961}\makebox(0,0)[lt]{\lineheight{1.25}\smash{\begin{tabular}[t]{l}20\end{tabular}}}}%
    \put(0.05638629,0.31542079){\color[rgb]{0.14901961,0.14901961,0.14901961}\rotatebox{90.00000248}{\makebox(0,0)[lt]{\lineheight{1.25}\smash{\begin{tabular}[t]{l}\normalsize $x_2$\end{tabular}}}}}%
    \put(0.52102858,0.58101636){\color[rgb]{0,0,0}\makebox(0,0)[lt]{\lineheight{1.25}\smash{\begin{tabular}[t]{l}\normalsize \!\!\!\!$\tau=5$\end{tabular}}}}%
    \put(0,0){\includegraphics[width=\unitlength,page=3]{gero_dt_5.pdf}}%
  \end{picture}%
\endgroup%

%% file: gero_dt_10.pdf_tex
\begingroup%
  \makeatletter%
  \providecommand\color[2][]{%
    \errmessage{(Inkscape) Color is used for the text in Inkscape, but the package 'color.sty' is not loaded}%
    \renewcommand\color[2][]{}%
  }%
  \providecommand\transparent[1]{%
    \errmessage{(Inkscape) Transparency is used (non-zero) for the text in Inkscape, but the package 'transparent.sty' is not loaded}%
    \renewcommand\transparent[1]{}%
  }%
  \providecommand\rotatebox[2]{#2}%
  \newcommand*\fsize{\dimexpr\f@size pt\relax}%
  \newcommand*\lineheight[1]{\fontsize{\fsize}{#1\fsize}\selectfont}%
  \ifx\svgwidth\undefined%
    \setlength{\unitlength}{160.49999599bp}%
    \ifx\svgscale\undefined%
      \relax%
    \else%
      \setlength{\unitlength}{\unitlength * \real{\svgscale}}%
    \fi%
  \else%
    \setlength{\unitlength}{\svgwidth}%
  \fi%
  \global\let\svgwidth\undefined%
  \global\let\svgscale\undefined%
  \makeatother \scriptsize%
  \begin{picture}(1,0.64018692)%
    \lineheight{1}%
    \setlength\tabcolsep{0pt}%
    \put(0,0){\includegraphics[width=\unitlength,page=1]{gero_dt_10.pdf}}%
    \put(0.15420561,0.08442368){\color[rgb]{0.14901961,0.14901961,0.14901961}\makebox(0,0)[lt]{\lineheight{1.25}\smash{\begin{tabular}[t]{l}-5\end{tabular}}}}%
    \put(0.47157897,0.08442368){\color[rgb]{0.14901961,0.14901961,0.14901961}\makebox(0,0)[lt]{\lineheight{1.25}\smash{\begin{tabular}[t]{l}0\end{tabular}}}}%
    \put(0.77960654,0.08442368){\color[rgb]{0.14901961,0.14901961,0.14901961}\makebox(0,0)[lt]{\lineheight{1.25}\smash{\begin{tabular}[t]{l}5\end{tabular}}}}%
    \put(0.50700974,0.01277259){\color[rgb]{0.14901961,0.14901961,0.14901961}\makebox(0,0)[lt]{\lineheight{1.25}\smash{\begin{tabular}[t]{l}\normalsize $x_1$\end{tabular}}}}%
    \put(0,0){\includegraphics[width=\unitlength,page=2]{gero_dt_10.pdf}}%
    \put(0.07975078,0.13084112){\color[rgb]{0.14901961,0.14901961,0.14901961}\makebox(0,0)[lt]{\lineheight{1.25}\smash{\begin{tabular}[t]{l}-20\end{tabular}}}}%
    \put(0.07975078,0.25533972){\color[rgb]{0.14901961,0.14901961,0.14901961}\makebox(0,0)[lt]{\lineheight{1.25}\smash{\begin{tabular}[t]{l}-10\end{tabular}}}}%
    \put(0.13115265,0.37983816){\color[rgb]{0.14901961,0.14901961,0.14901961}\makebox(0,0)[lt]{\lineheight{1.25}\smash{\begin{tabular}[t]{l}0\end{tabular}}}}%
    \put(0.09844237,0.50433668){\color[rgb]{0.14901961,0.14901961,0.14901961}\makebox(0,0)[lt]{\lineheight{1.25}\smash{\begin{tabular}[t]{l}10\end{tabular}}}}%
    \put(0.05638629,0.31542079){\color[rgb]{0.14901961,0.14901961,0.14901961}\rotatebox{90.00000248}{\makebox(0,0)[lt]{\lineheight{1.25}\smash{\begin{tabular}[t]{l}\normalsize $x_2$\end{tabular}}}}}%
    \put(0.52102858,0.58101643){\color[rgb]{0,0,0}\makebox(0,0)[lt]{\lineheight{1.25}\smash{\begin{tabular}[t]{l}\normalsize \!\!\!\!$\tau=10$\end{tabular}}}}%
    \put(0,0){\includegraphics[width=\unitlength,page=3]{gero_dt_10.pdf}}%
  \end{picture}%
\endgroup%

%% file: ci_opt_iter.pdf_tex
\begingroup%
  \makeatletter%
  \providecommand\color[2][]{%
    \errmessage{(Inkscape) Color is used for the text in Inkscape, but the package 'color.sty' is not loaded}%
    \renewcommand\color[2][]{}%
  }%
  \providecommand\transparent[1]{%
    \errmessage{(Inkscape) Transparency is used (non-zero) for the text in Inkscape, but the package 'transparent.sty' is not loaded}%
    \renewcommand\transparent[1]{}%
  }%
  \providecommand\rotatebox[2]{#2}%
  \newcommand*\fsize{\dimexpr\f@size pt\relax}%
  \newcommand*\lineheight[1]{\fontsize{\fsize}{#1\fsize}\selectfont}%
  \ifx\svgwidth\undefined%
    \setlength{\unitlength}{195.74999511bp}%
    \ifx\svgscale\undefined%
      \relax%
    \else%
      \setlength{\unitlength}{\unitlength * \real{\svgscale}}%
    \fi%
  \else%
    \setlength{\unitlength}{\svgwidth}%
  \fi%
  \global\let\svgwidth\undefined%
  \global\let\svgscale\undefined%
  \makeatother \scriptsize%
  \begin{picture}(1,0.44636016)%
    \lineheight{1}%
    \setlength\tabcolsep{0pt}%
    \put(0,0){\includegraphics[width=\unitlength,page=1]{ci_opt_iter.pdf}}%
    \put(0.15900383,0.06816208){\color[rgb]{0.14901961,0.14901961,0.14901961}\makebox(0,0)[lt]{\lineheight{1.25}\smash{\begin{tabular}[t]{l}0\end{tabular}}}}%
    \put(0.29195402,0.06816208){\color[rgb]{0.14901961,0.14901961,0.14901961}\makebox(0,0)[lt]{\lineheight{1.25}\smash{\begin{tabular}[t]{l}10\end{tabular}}}}%
    \put(0.43831418,0.06816208){\color[rgb]{0.14901961,0.14901961,0.14901961}\makebox(0,0)[lt]{\lineheight{1.25}\smash{\begin{tabular}[t]{l}20\end{tabular}}}}%
    \put(0.58467433,0.06816208){\color[rgb]{0.14901961,0.14901961,0.14901961}\makebox(0,0)[lt]{\lineheight{1.25}\smash{\begin{tabular}[t]{l}30\end{tabular}}}}%
    \put(0.73103448,0.06816208){\color[rgb]{0.14901961,0.14901961,0.14901961}\makebox(0,0)[lt]{\lineheight{1.25}\smash{\begin{tabular}[t]{l}40\end{tabular}}}}%
    \put(0.87739464,0.06816208){\color[rgb]{0.14901961,0.14901961,0.14901961}\makebox(0,0)[lt]{\lineheight{1.25}\smash{\begin{tabular}[t]{l}50\end{tabular}}}}%
    \put(0.50000032,0.01526012){\color[rgb]{0.14901961,0.14901961,0.14901961}\makebox(0,0)[lt]{\lineheight{1.25}\smash{\begin{tabular}[t]{l}\normalsize \!\!\!\!\!Iteration\end{tabular}}}}%
    \put(0,0){\includegraphics[width=\unitlength,page=2]{ci_opt_iter.pdf}}%
    \put(0.08378033,0.10109862){\color[rgb]{0.14901961,0.14901961,0.14901961}\makebox(0,0)[lt]{\lineheight{1.25}\smash{\begin{tabular}[t]{l}500\end{tabular}}}}%
    \put(0.05696041,0.17549198){\color[rgb]{0.14901961,0.14901961,0.14901961}\makebox(0,0)[lt]{\lineheight{1.25}\smash{\begin{tabular}[t]{l}1000\end{tabular}}}}%
    \put(0.05696041,0.24988534){\color[rgb]{0.14901961,0.14901961,0.14901961}\makebox(0,0)[lt]{\lineheight{1.25}\smash{\begin{tabular}[t]{l}1500\end{tabular}}}}%
    \put(0.05696041,0.3242787){\color[rgb]{0.14901961,0.14901961,0.14901961}\makebox(0,0)[lt]{\lineheight{1.25}\smash{\begin{tabular}[t]{l}2000\end{tabular}}}}%
    \put(0.05696041,0.39867206){\color[rgb]{0.14901961,0.14901961,0.14901961}\makebox(0,0)[lt]{\lineheight{1.25}\smash{\begin{tabular}[t]{l}2500\end{tabular}}}}%
    \put(0.0421456,0.13734173){\color[rgb]{0.14901961,0.14901961,0.14901961}\rotatebox{90.00000249}{\makebox(0,0)[lt]{\lineheight{1.25}\smash{\begin{tabular}[t]{l}\normalsize \!Obj. Fun.\end{tabular}}}}}%
    \put(0,0){\includegraphics[width=\unitlength,page=3]{ci_opt_iter.pdf}}%
  \end{picture}%
\endgroup%

%% file: ci_opt_iter_2.pdf_tex
\begingroup%
  \makeatletter%
  \providecommand\color[2][]{%
    \errmessage{(Inkscape) Color is used for the text in Inkscape, but the package 'color.sty' is not loaded}%
    \renewcommand\color[2][]{}%
  }%
  \providecommand\transparent[1]{%
    \errmessage{(Inkscape) Transparency is used (non-zero) for the text in Inkscape, but the package 'transparent.sty' is not loaded}%
    \renewcommand\transparent[1]{}%
  }%
  \providecommand\rotatebox[2]{#2}%
  \newcommand*\fsize{\dimexpr\f@size pt\relax}%
  \newcommand*\lineheight[1]{\fontsize{\fsize}{#1\fsize}\selectfont}%
  \ifx\svgwidth\undefined%
    \setlength{\unitlength}{203.24999492bp}%
    \ifx\svgscale\undefined%
      \relax%
    \else%
      \setlength{\unitlength}{\unitlength * \real{\svgscale}}%
    \fi%
  \else%
    \setlength{\unitlength}{\svgwidth}%
  \fi%
  \global\let\svgwidth\undefined%
  \global\let\svgscale\undefined%
  \makeatother \scriptsize%
  \begin{picture}(1,0.40590406)%
    \lineheight{1}%
    \setlength\tabcolsep{0pt}%
    \put(0,0){\includegraphics[width=\unitlength,page=1]{ci_opt_iter_2.pdf}}%
    \put(0.11623616,0.06223868){\color[rgb]{0.14901961,0.14901961,0.14901961}\makebox(0,0)[lt]{\lineheight{1.25}\smash{\begin{tabular}[t]{l}0\end{tabular}}}}%
    \put(0.25830258,0.06223868){\color[rgb]{0.14901961,0.14901961,0.14901961}\makebox(0,0)[lt]{\lineheight{1.25}\smash{\begin{tabular}[t]{l}10\end{tabular}}}}%
    \put(0.41328413,0.06223868){\color[rgb]{0.14901961,0.14901961,0.14901961}\makebox(0,0)[lt]{\lineheight{1.25}\smash{\begin{tabular}[t]{l}20\end{tabular}}}}%
    \put(0.56826568,0.06223868){\color[rgb]{0.14901961,0.14901961,0.14901961}\makebox(0,0)[lt]{\lineheight{1.25}\smash{\begin{tabular}[t]{l}30\end{tabular}}}}%
    \put(0.72324723,0.06223868){\color[rgb]{0.14901961,0.14901961,0.14901961}\makebox(0,0)[lt]{\lineheight{1.25}\smash{\begin{tabular}[t]{l}40\end{tabular}}}}%
    \put(0.87822878,0.06223868){\color[rgb]{0.14901961,0.14901961,0.14901961}\makebox(0,0)[lt]{\lineheight{1.25}\smash{\begin{tabular}[t]{l}50\end{tabular}}}}%
    \put(0.4797051,0.01107011){\color[rgb]{0.14901961,0.14901961,0.14901961}\makebox(0,0)[lt]{\lineheight{1.25}\smash{\begin{tabular}[t]{l}\normalsize \!\!\!\!\!Iteration\end{tabular}}}}%
    \put(0,0){\includegraphics[width=\unitlength,page=2]{ci_opt_iter_2.pdf}}%
    \put(0.0696187,0.1082377){\color[rgb]{0.14901961,0.14901961,0.14901961}\makebox(0,0)[lt]{\lineheight{1.25}\smash{\begin{tabular}[t]{l}15\end{tabular}}}}%
    \put(0.0696187,0.1933147){\color[rgb]{0.14901961,0.14901961,0.14901961}\makebox(0,0)[lt]{\lineheight{1.25}\smash{\begin{tabular}[t]{l}20\end{tabular}}}}%
    \put(0.0696187,0.2783917){\color[rgb]{0.14901961,0.14901961,0.14901961}\makebox(0,0)[lt]{\lineheight{1.25}\smash{\begin{tabular}[t]{l}25\end{tabular}}}}%
    \put(0.0696187,0.36346863){\color[rgb]{0.14901961,0.14901961,0.14901961}\makebox(0,0)[lt]{\lineheight{1.25}\smash{\begin{tabular}[t]{l}30\end{tabular}}}}%
    \put(0.05362854,0.11992632){\color[rgb]{0.14901961,0.14901961,0.14901961}\rotatebox{90.00000248}{\makebox(0,0)[lt]{\lineheight{1.25}\smash{\begin{tabular}[t]{l}\normalsize \!Obj. Fun.\end{tabular}}}}}%
    \put(0,0){\includegraphics[width=\unitlength,page=3]{ci_opt_iter_2.pdf}}%
  \end{picture}%
\endgroup%

%% file: ci_dt_xe.pdf_tex
\begingroup%
  \makeatletter%
  \providecommand\color[2][]{%
    \errmessage{(Inkscape) Color is used for the text in Inkscape, but the package 'color.sty' is not loaded}%
    \renewcommand\color[2][]{}%
  }%
  \providecommand\transparent[1]{%
    \errmessage{(Inkscape) Transparency is used (non-zero) for the text in Inkscape, but the package 'transparent.sty' is not loaded}%
    \renewcommand\transparent[1]{}%
  }%
  \providecommand\rotatebox[2]{#2}%
  \newcommand*\fsize{\dimexpr\f@size pt\relax}%
  \newcommand*\lineheight[1]{\fontsize{\fsize}{#1\fsize}\selectfont}%
  \ifx\svgwidth\undefined%
    \setlength{\unitlength}{185.99999535bp}%
    \ifx\svgscale\undefined%
      \relax%
    \else%
      \setlength{\unitlength}{\unitlength * \real{\svgscale}}%
    \fi%
  \else%
    \setlength{\unitlength}{\svgwidth}%
  \fi%
  \global\let\svgwidth\undefined%
  \global\let\svgscale\undefined%
  \makeatother \scriptsize%
  \begin{picture}(1,0.64919355)%
    \lineheight{1}%
    \setlength\tabcolsep{0pt}%
    \put(0,0){\includegraphics[width=\unitlength,page=1]{ci_dt_xe.pdf}}%
    \put(0.11895161,0.07284946){\color[rgb]{0.14901961,0.14901961,0.14901961}\makebox(0,0)[lt]{\lineheight{1.25}\smash{\begin{tabular}[t]{l}-50\end{tabular}}}}%
    \put(0.5141129,0.07284946){\color[rgb]{0.14901961,0.14901961,0.14901961}\makebox(0,0)[lt]{\lineheight{1.25}\smash{\begin{tabular}[t]{l}0\end{tabular}}}}%
    \put(0.87298387,0.07284946){\color[rgb]{0.14901961,0.14901961,0.14901961}\makebox(0,0)[lt]{\lineheight{1.25}\smash{\begin{tabular}[t]{l}50\end{tabular}}}}%
    \put(0.49798421,0.01102157){\color[rgb]{0.14901961,0.14901961,0.14901961}\makebox(0,0)[lt]{\lineheight{1.25}\smash{\begin{tabular}[t]{l}\normalsize $x_1$\end{tabular}}}}%
    \put(0,0){\includegraphics[width=\unitlength,page=2]{ci_dt_xe.pdf}}%
    \put(0.0688172,0.11290323){\color[rgb]{0.14901961,0.14901961,0.14901961}\makebox(0,0)[lt]{\lineheight{1.25}\smash{\begin{tabular}[t]{l}-60\end{tabular}}}}%
    \put(0.0688172,0.20241935){\color[rgb]{0.14901961,0.14901961,0.14901961}\makebox(0,0)[lt]{\lineheight{1.25}\smash{\begin{tabular}[t]{l}-40\end{tabular}}}}%
    \put(0.0688172,0.29193555){\color[rgb]{0.14901961,0.14901961,0.14901961}\makebox(0,0)[lt]{\lineheight{1.25}\smash{\begin{tabular}[t]{l}-20\end{tabular}}}}%
    \put(0.11317204,0.38145161){\color[rgb]{0.14901961,0.14901961,0.14901961}\makebox(0,0)[lt]{\lineheight{1.25}\smash{\begin{tabular}[t]{l}0\end{tabular}}}}%
    \put(0.08494624,0.47096774){\color[rgb]{0.14901961,0.14901961,0.14901961}\makebox(0,0)[lt]{\lineheight{1.25}\smash{\begin{tabular}[t]{l}20\end{tabular}}}}%
    \put(0.08494624,0.56048387){\color[rgb]{0.14901961,0.14901961,0.14901961}\makebox(0,0)[lt]{\lineheight{1.25}\smash{\begin{tabular}[t]{l}40\end{tabular}}}}%
    \put(0.04865592,0.32056478){\color[rgb]{0.14901961,0.14901961,0.14901961}\rotatebox{90.00000248}{\makebox(0,0)[lt]{\lineheight{1.25}\smash{\begin{tabular}[t]{l}\normalsize $x_2$\end{tabular}}}}}%
    \put(0.51008112,0.59813508){\color[rgb]{0,0,0}\makebox(0,0)[lt]{\lineheight{1.25}\smash{\begin{tabular}[t]{l}\normalsize \!\!\!\!$c_i=x_{ei}$\end{tabular}}}}%
    \put(0,0){\includegraphics[width=\unitlength,page=3]{ci_dt_xe.pdf}}%
  \end{picture}%
\endgroup%

%% file: ci_dt_opt.pdf_tex
\begingroup%
  \makeatletter%
  \providecommand\color[2][]{%
    \errmessage{(Inkscape) Color is used for the text in Inkscape, but the package 'color.sty' is not loaded}%
    \renewcommand\color[2][]{}%
  }%
  \providecommand\transparent[1]{%
    \errmessage{(Inkscape) Transparency is used (non-zero) for the text in Inkscape, but the package 'transparent.sty' is not loaded}%
    \renewcommand\transparent[1]{}%
  }%
  \providecommand\rotatebox[2]{#2}%
  \newcommand*\fsize{\dimexpr\f@size pt\relax}%
  \newcommand*\lineheight[1]{\fontsize{\fsize}{#1\fsize}\selectfont}%
  \ifx\svgwidth\undefined%
    \setlength{\unitlength}{185.99999535bp}%
    \ifx\svgscale\undefined%
      \relax%
    \else%
      \setlength{\unitlength}{\unitlength * \real{\svgscale}}%
    \fi%
  \else%
    \setlength{\unitlength}{\svgwidth}%
  \fi%
  \global\let\svgwidth\undefined%
  \global\let\svgscale\undefined%
  \makeatother \scriptsize%
  \begin{picture}(1,0.64919355)%
    \lineheight{1}%
    \setlength\tabcolsep{0pt}%
    \put(0,0){\includegraphics[width=\unitlength,page=1]{ci_dt_opt.pdf}}%
    \put(0.11895161,0.07284946){\color[rgb]{0.14901961,0.14901961,0.14901961}\makebox(0,0)[lt]{\lineheight{1.25}\smash{\begin{tabular}[t]{l}-50\end{tabular}}}}%
    \put(0.5141129,0.07284946){\color[rgb]{0.14901961,0.14901961,0.14901961}\makebox(0,0)[lt]{\lineheight{1.25}\smash{\begin{tabular}[t]{l}0\end{tabular}}}}%
    \put(0.87298387,0.07284946){\color[rgb]{0.14901961,0.14901961,0.14901961}\makebox(0,0)[lt]{\lineheight{1.25}\smash{\begin{tabular}[t]{l}50\end{tabular}}}}%
    \put(0.49798421,0.01102157){\color[rgb]{0.14901961,0.14901961,0.14901961}\makebox(0,0)[lt]{\lineheight{1.25}\smash{\begin{tabular}[t]{l}\normalsize $x_1$\end{tabular}}}}%
    \put(0,0){\includegraphics[width=\unitlength,page=2]{ci_dt_opt.pdf}}%
    \put(0.0688172,0.11290323){\color[rgb]{0.14901961,0.14901961,0.14901961}\makebox(0,0)[lt]{\lineheight{1.25}\smash{\begin{tabular}[t]{l}-60\end{tabular}}}}%
    \put(0.0688172,0.20241935){\color[rgb]{0.14901961,0.14901961,0.14901961}\makebox(0,0)[lt]{\lineheight{1.25}\smash{\begin{tabular}[t]{l}-40\end{tabular}}}}%
    \put(0.0688172,0.29193555){\color[rgb]{0.14901961,0.14901961,0.14901961}\makebox(0,0)[lt]{\lineheight{1.25}\smash{\begin{tabular}[t]{l}-20\end{tabular}}}}%
    \put(0.11317204,0.38145161){\color[rgb]{0.14901961,0.14901961,0.14901961}\makebox(0,0)[lt]{\lineheight{1.25}\smash{\begin{tabular}[t]{l}0\end{tabular}}}}%
    \put(0.08494624,0.47096774){\color[rgb]{0.14901961,0.14901961,0.14901961}\makebox(0,0)[lt]{\lineheight{1.25}\smash{\begin{tabular}[t]{l}20\end{tabular}}}}%
    \put(0.08494624,0.56048387){\color[rgb]{0.14901961,0.14901961,0.14901961}\makebox(0,0)[lt]{\lineheight{1.25}\smash{\begin{tabular}[t]{l}40\end{tabular}}}}%
    \put(0.04865592,0.32056478){\color[rgb]{0.14901961,0.14901961,0.14901961}\rotatebox{90.00000248}{\makebox(0,0)[lt]{\lineheight{1.25}\smash{\begin{tabular}[t]{l}\normalsize $x_2$\end{tabular}}}}}%
    \put(0.51008112,0.59813508){\color[rgb]{0,0,0}\makebox(0,0)[lt]{\lineheight{1.25}\smash{\begin{tabular}[t]{l}\normalsize \!\!\!\!$c_i$ optimized\end{tabular}}}}%
    \put(0,0){\includegraphics[width=\unitlength,page=3]{ci_dt_opt.pdf}}%
  \end{picture}%
\endgroup%

%% file: ci_dt_xe_2.pdf_tex
\begingroup%
  \makeatletter%
  \providecommand\color[2][]{%
    \errmessage{(Inkscape) Color is used for the text in Inkscape, but the package 'color.sty' is not loaded}%
    \renewcommand\color[2][]{}%
  }%
  \providecommand\transparent[1]{%
    \errmessage{(Inkscape) Transparency is used (non-zero) for the text in Inkscape, but the package 'transparent.sty' is not loaded}%
    \renewcommand\transparent[1]{}%
  }%
  \providecommand\rotatebox[2]{#2}%
  \newcommand*\fsize{\dimexpr\f@size pt\relax}%
  \newcommand*\lineheight[1]{\fontsize{\fsize}{#1\fsize}\selectfont}%
  \ifx\svgwidth\undefined%
    \setlength{\unitlength}{185.99999535bp}%
    \ifx\svgscale\undefined%
      \relax%
    \else%
      \setlength{\unitlength}{\unitlength * \real{\svgscale}}%
    \fi%
  \else%
    \setlength{\unitlength}{\svgwidth}%
  \fi%
  \global\let\svgwidth\undefined%
  \global\let\svgscale\undefined%
  \makeatother \scriptsize%
  \begin{picture}(1,0.64919355)%
    \lineheight{1}%
    \setlength\tabcolsep{0pt}%
    \put(0,0){\includegraphics[width=\unitlength,page=1]{ci_dt_xe_2.pdf}}%
    \put(0.21219758,0.07284946){\color[rgb]{0.14901961,0.14901961,0.14901961}\makebox(0,0)[lt]{\lineheight{1.25}\smash{\begin{tabular}[t]{l}-10\end{tabular}}}}%
    \put(0.4594254,0.07284946){\color[rgb]{0.14901961,0.14901961,0.14901961}\makebox(0,0)[lt]{\lineheight{1.25}\smash{\begin{tabular}[t]{l}-5\end{tabular}}}}%
    \put(0.70060484,0.07284946){\color[rgb]{0.14901961,0.14901961,0.14901961}\makebox(0,0)[lt]{\lineheight{1.25}\smash{\begin{tabular}[t]{l}0\end{tabular}}}}%
    \put(0.49798421,0.01102157){\color[rgb]{0.14901961,0.14901961,0.14901961}\makebox(0,0)[lt]{\lineheight{1.25}\smash{\begin{tabular}[t]{l}\normalsize $x_1$\end{tabular}}}}%
    \put(0,0){\includegraphics[width=\unitlength,page=2]{ci_dt_xe_2.pdf}}%
    \put(0.0688172,0.11290323){\color[rgb]{0.14901961,0.14901961,0.14901961}\makebox(0,0)[lt]{\lineheight{1.25}\smash{\begin{tabular}[t]{l}-15\end{tabular}}}}%
    \put(0.0688172,0.22479839){\color[rgb]{0.14901961,0.14901961,0.14901961}\makebox(0,0)[lt]{\lineheight{1.25}\smash{\begin{tabular}[t]{l}-10\end{tabular}}}}%
    \put(0.09704301,0.33669355){\color[rgb]{0.14901961,0.14901961,0.14901961}\makebox(0,0)[lt]{\lineheight{1.25}\smash{\begin{tabular}[t]{l}-5\end{tabular}}}}%
    \put(0.11317204,0.44858871){\color[rgb]{0.14901961,0.14901961,0.14901961}\makebox(0,0)[lt]{\lineheight{1.25}\smash{\begin{tabular}[t]{l}0\end{tabular}}}}%
    \put(0.11317204,0.56048387){\color[rgb]{0.14901961,0.14901961,0.14901961}\makebox(0,0)[lt]{\lineheight{1.25}\smash{\begin{tabular}[t]{l}5\end{tabular}}}}%
    \put(0.04865592,0.32056478){\color[rgb]{0.14901961,0.14901961,0.14901961}\rotatebox{90.00000248}{\makebox(0,0)[lt]{\lineheight{1.25}\smash{\begin{tabular}[t]{l}\normalsize $x_2$\end{tabular}}}}}%
    \put(0.51008112,0.59813508){\color[rgb]{0,0,0}\makebox(0,0)[lt]{\lineheight{1.25}\smash{\begin{tabular}[t]{l}\normalsize \!\!\!\!$c_i=x_{ei}$\end{tabular}}}}%
    \put(0,0){\includegraphics[width=\unitlength,page=3]{ci_dt_xe_2.pdf}}%
  \end{picture}%
\endgroup%

%% file: ci_dt_opt_2.pdf_tex
\begingroup%
  \makeatletter%
  \providecommand\color[2][]{%
    \errmessage{(Inkscape) Color is used for the text in Inkscape, but the package 'color.sty' is not loaded}%
    \renewcommand\color[2][]{}%
  }%
  \providecommand\transparent[1]{%
    \errmessage{(Inkscape) Transparency is used (non-zero) for the text in Inkscape, but the package 'transparent.sty' is not loaded}%
    \renewcommand\transparent[1]{}%
  }%
  \providecommand\rotatebox[2]{#2}%
  \newcommand*\fsize{\dimexpr\f@size pt\relax}%
  \newcommand*\lineheight[1]{\fontsize{\fsize}{#1\fsize}\selectfont}%
  \ifx\svgwidth\undefined%
    \setlength{\unitlength}{185.99999535bp}%
    \ifx\svgscale\undefined%
      \relax%
    \else%
      \setlength{\unitlength}{\unitlength * \real{\svgscale}}%
    \fi%
  \else%
    \setlength{\unitlength}{\svgwidth}%
  \fi%
  \global\let\svgwidth\undefined%
  \global\let\svgscale\undefined%
  \makeatother \scriptsize%
  \begin{picture}(1,0.64919355)%
    \lineheight{1}%
    \setlength\tabcolsep{0pt}%
    \put(0,0){\includegraphics[width=\unitlength,page=1]{ci_dt_opt_2.pdf}}%
    \put(0.21219758,0.07284946){\color[rgb]{0.14901961,0.14901961,0.14901961}\makebox(0,0)[lt]{\lineheight{1.25}\smash{\begin{tabular}[t]{l}-10\end{tabular}}}}%
    \put(0.4594254,0.07284946){\color[rgb]{0.14901961,0.14901961,0.14901961}\makebox(0,0)[lt]{\lineheight{1.25}\smash{\begin{tabular}[t]{l}-5\end{tabular}}}}%
    \put(0.70060484,0.07284946){\color[rgb]{0.14901961,0.14901961,0.14901961}\makebox(0,0)[lt]{\lineheight{1.25}\smash{\begin{tabular}[t]{l}0\end{tabular}}}}%
    \put(0.49798421,0.01102157){\color[rgb]{0.14901961,0.14901961,0.14901961}\makebox(0,0)[lt]{\lineheight{1.25}\smash{\begin{tabular}[t]{l}\normalsize $x_1$\end{tabular}}}}%
    \put(0,0){\includegraphics[width=\unitlength,page=2]{ci_dt_opt_2.pdf}}%
    \put(0.0688172,0.11290323){\color[rgb]{0.14901961,0.14901961,0.14901961}\makebox(0,0)[lt]{\lineheight{1.25}\smash{\begin{tabular}[t]{l}-15\end{tabular}}}}%
    \put(0.0688172,0.22479839){\color[rgb]{0.14901961,0.14901961,0.14901961}\makebox(0,0)[lt]{\lineheight{1.25}\smash{\begin{tabular}[t]{l}-10\end{tabular}}}}%
    \put(0.09704301,0.33669355){\color[rgb]{0.14901961,0.14901961,0.14901961}\makebox(0,0)[lt]{\lineheight{1.25}\smash{\begin{tabular}[t]{l}-5\end{tabular}}}}%
    \put(0.11317204,0.44858871){\color[rgb]{0.14901961,0.14901961,0.14901961}\makebox(0,0)[lt]{\lineheight{1.25}\smash{\begin{tabular}[t]{l}0\end{tabular}}}}%
    \put(0.11317204,0.56048387){\color[rgb]{0.14901961,0.14901961,0.14901961}\makebox(0,0)[lt]{\lineheight{1.25}\smash{\begin{tabular}[t]{l}5\end{tabular}}}}%
    \put(0.04865592,0.32056478){\color[rgb]{0.14901961,0.14901961,0.14901961}\rotatebox{90.00000248}{\makebox(0,0)[lt]{\lineheight{1.25}\smash{\begin{tabular}[t]{l}\normalsize $x_2$\end{tabular}}}}}%
    \put(0.51008112,0.59813508){\color[rgb]{0,0,0}\makebox(0,0)[lt]{\lineheight{1.25}\smash{\begin{tabular}[t]{l}\normalsize \!\!\!\!$c_i$ optimized\end{tabular}}}}%
    \put(0,0){\includegraphics[width=\unitlength,page=3]{ci_dt_opt_2.pdf}}%
  \end{picture}%
\endgroup%